\documentclass[11pt]{article}
\usepackage{amsmath, amsthm, amssymb,a4wide}
\usepackage{enumitem}

\newtheorem{theorem}{Theorem}[section]
\newtheorem*{theorem*}{Theorem}

\newtheorem{lemma}[theorem]{Lemma}
\newtheorem*{lemma*}{Lemma}

\newtheorem{proposition}[theorem]{Proposition}

\newtheorem{definition}[theorem]{Definition}
\newtheorem*{definition*}{Definition}

\theoremstyle{remark}
\newtheorem{remark}{Remark}
\newtheorem*{remark*}{Remark}

\usepackage{tikz}
\usetikzlibrary{intersections,positioning,shapes,calc,arrows,decorations.pathmorphing,decorations.pathreplacing}

\tikzset{iNode/.style={draw=blue, rectangle}}
\tikzset{fNode/.style={draw=green, circle}}
\tikzset{rNode/.style={draw=red, circle}}
\tikzset{nNode/.style={draw, circle}}

\newcommand{\prob}{\Pr}

\def\zo{\{0,1\}}
\def\mapping{\rightarrow}
\newcommand{\ie}{$\mbox{i.e.}$}
\newcommand{\nat}{{\mathbb N}}

\newcommand{\mctX}{\widetilde{\mathcal{X}}}
\newcommand{\mcX}{\mathcal{X}}

\newcommand{\mcS}{\mathcal{S}}
\newcommand{\mcR}{\mathcal{R}}
\newcommand{\mcM}{\mathcal{M}}

\newcommand{\E}{\mathcal{E}}

\newcommand{\tx}{{\tilde{x}}}

\newcommand{\ttd}{{\cal D}}  
\newcommand{\ttc}{{\cal C}}

\newcommand{\enc}{\mathrm{Enc}}
\newcommand{\dec}{\mathrm{Dec}}
\newcommand{\cha}{\mathrm{Ch}}
\newcommand{\Ch}{\cha}
\newcommand{\tk}{t}
\newcommand{\kt}{k}
\newcommand{\guv}{\rm{GUV}}
\newcommand{\tu}{\rm{TU}}
\newcommand{\rrv}{\rm{RRV}}
\newcommand{\tm}{t_{\rm{max}}}
\newcommand{\F}{\mathbb{F}}

\newcommand{\eps}{\varepsilon}

\renewcommand{\mid}{\,|\,}

\usepackage{color,soul}
\definecolor{lightblue}{rgb}{.60,.60,1}
\definecolor{lightred}{rgb}{1, .60, 0.60}
\soulregister\ref7

\newcommand{\marius}[1]{\color{red} #1 \color{black}}

\makeatletter
\def\@listI{\leftmargin\leftmargini \parsep 4.5pt plus 1pt minus 1pt\topsep6pt plus 2pt minus 2pt \itemsep  2pt plus 2pt minus 1pt}

\let\@listi\@listI
\@listi
\makeatother

\title{Universal codes in the  shared-randomness model   for  channels with general distortion capabilities}

\author{
{Bruno Bauwens\/}
\thanks{National Research University Higher School of Economics, Faculty of Computer Science, Moscow, Russia, email:~\texttt{brbauwens@gmail.com};
Chapters 3 and 4 were prepared at the National Research University Higher School of Economics (HSE University) and supported by Russian Science Foundation (grant 20-11-20203).
}
\quad{Marius Zimand\/}
\thanks{  Department of Computer and Information Sciences, Towson University,
Baltimore, MD. http://orion.towson.edu/\~{ }mzimand ; Chapters 1, 2, 5 and appendix A were prepared at Towson University. The author has been supported in part by the National Science Foundation through grant CCF 1811729.}}

\begin{document}
 
\date{}

\maketitle
\begin{abstract}
We put forth new models  for  universal channel coding.
Unlike standard codes which are designed for a specific type of channel,  our most general universal  code makes  communication resilient on every channel,  
 provided the noise level is below the tolerated bound, where the noise level $t$ of a channel is the logarithm of its ambiguity (the maximum number of strings that can be distorted into a given one). The other  more restricted universal codes that we introduce  still work for large classes of natural channels. 
  In a universal code,  encoding  is channel-independent,  but  the decoding function knows the type of channel.  We allow the encoding and the decoding functions to share randomness, which is unavailable to the channel.  There are two scenarios for the type of attack that a channel can perform. In the oblivious scenario, 
  codewords belong to an additive group and the channel distorts a codeword by adding a vector from a fixed set. The selection is
  based on the message and the encoding function, but not on the codeword. 
  In the Hamming scenario, the channel knows the codeword and is fully adversarial.   For a universal code, there are two parameters of interest: the rate, which is the ratio between the message length $k$ and the codeword length $n$, and the number of shared random bits.  We show the existence in both scenarios  of universal codes with rate $1-t/n - o(1)$,  which is optimal modulo the $o(1)$ term.  The number of shared random bits is $O(\log n)$ in the oblivious scenario, and $O(n)$  in the Hamming scenario, which, for typical values of the noise level,  we show to be optimal, modulo the constant hidden in the $O(\cdot)$ notation.  In both scenarios,  the universal encoding is done in time polynomial in $n$, but  the channel-dependent decoding procedures are  in general not efficient. For some weaker classes of channels which produce the distortion based on short blocks of the codeword (rather than the entire codeword), we construct universal codes with polynomial-time encoding and decoding. Furthermore, for channels that work in the memoryless oblivious scenario, where they choose the noise vector  for each block randomly (rather than adversarially) and independently, there exists a universal code with deterministic polynomial-time encoding/decoding.
\end{abstract}

\newpage



\section{Introduction}

In the problem of  {\em channel coding} a sender  needs to communicate data  over a noisy channel to a receiver. 
In the most general setting a message $m$ is encoded into a codeword $x$. This codeword is transmitted over the channel who distorts $x$ into $\tx$. Then a decoder tries to reconstruct $m$ from~$\tx$.
\[
 m  \xrightarrow[\text{\quad Encoder \quad }]{} x  \xrightarrow[\text{\quad \quad  \quad Channel \quad \quad \quad }]{} \tx \xrightarrow[\text{\quad Decoder \quad}]{} m
\]

The channel is viewed as an adversary and is characterized by the type of operations it uses to produce the noise, and by a parameter $\tk$, which quantitatively describes the maximum noise that we want to tolerate. 
Roughly speaking, most studies have focused on channels  defined by a fixed set ${\cal T}$ of possible operations  that add noise and by setting $\tk$ to be  the maximum  number of operations from ${\cal T}$ that the (encoder, decoder) pair can handle. 
Perhaps the most investigated setting  is the theory of error-correcting codes, where ${\cal T}$ consists of the single operation of $1$-bit flip ($0 \mapsto1, 1 \mapsto 0$).  In this case, $\tk$ is the maximum Hamming distance between $x$ and $\tx$ that is tolerated.   Another case where there has recently been significant progress is when ${\cal T}$ consists of the operations of $1$-bit flip, $1$-bit deletion, and $1$-bit insertion.   In this case, $\tk$ is the maximum edit distance between $x$ and $\tx$ that is tolerated.  Still another case is when ${\cal T}$ consists of the erasure operation which  transforms a bit into ``?". Other types of  channel that distort in various ways have also been investigated.

\medskip
Our setting is different in two important ways. 
Firstly, we consider channels that can do \emph{arbitrary} distortion. We consider two different scenarios on how the channel does the distortion, depending on whether it ``knows'' the codeword or only the message. 

\begin{itemize}[leftmargin=0pt,itemindent=*]
  \item[--] 
    In the {\em Hamming scenario} a channel is defined by a bipartite graph where left nodes represent codewords that are inputs of the channel, and right nodes represent distorted codewords that are outputs. A left and a right element are connected, if the channel may distort the codeword at the left, to the one at the right. 
    The level $t$ of the noise in the channel is the logarithm of the maximal degree of a right node, \ie, the logarithm of  the maximal number of input codewords of the channel that can produce the same distorted codeword. All left degrees are at least $1$.
No other assumptions are made on the channel.

  \item[--] 
In the {\em oblivious scenario} 
a channel takes as input codewords from an additive group. The channel is defined by a set of error vectors. On input a codeword for a message, it will add a vector from this set to the codeword. 
The choice of the error vector does not depend on the codeword, but on the message. The level $t$ of the noise is the logarithm of the size of the set.
\end{itemize}

Secondly,  our goal for each scenario is to have a single encoding function that is channel-independent. We call this a \emph{universal code.}  Differently said, a universal code is resilient to any type of distortion, provided the noise level is within the tolerated bound. On the other hand, for every channel there is a corresponding decoding function.\footnote{For any given decoding function. one can construct a channel that defeats it, and thus, unlike the encoding function,   it is impossible to  have a single decoding function  for all channels as well.  Anyway, decoding happens after the channel attack, and typically in  theoretical and practical applications,  the decoder knows the type of attack, or, at least, has a few candidates for it.} 

In order to construct universal codes, we  assume a special set-up for the communication process:  the universal encoder and the decoder functions are probabilistic and  share random bits. Such codes are called \emph{private codes}. They  have been introduced by Shannon~\cite{sha:j:private} (under the name random codes), and more recently studied by Langberg~\cite{lan:c:privatecode} (see also~\cite{smith2007scrambling, guruswami2016optimal}). 
%
The channel does not have access to the random shared bits, although, in the Hamming setting, the codeword might reveal some information about the randomness indirectly.

There are two important parameters. 
The first is the \emph{rate} of the code, which is defined by $\log K / \log N$, where $K$ is the number of messages that we can send, and $N$ is the number of codewords that the channel can transmit.  The second is the \emph{shared randomness} of the code, which is the number of random bits that the encoder and decoder share.  Given a noise level $t=\log T$, we want to maximize the rate and minimize the shared randomness.

It is not difficult to show that for a universal code, the value of the product $KT$ can not be larger than $N/(1-\epsilon)$, where $\epsilon$ is the error probability of the reconstruction of the message. 
This implies that the rate of such a code is at most $1- t/n - o(1)$, where $n = \log N$, see section~\ref{s:ub}.  
We construct universal codes with rates that converge to the optimal value and have small shared randomness.   The following simplified statements are valid for constant probability error. 

\begin{theorem}[Main Result - informal statement]\label{t:maininformal}
  \begin{itemize}[leftmargin=2em]
    \item[]
    \item[(a)] 
      There exists a universal code in the Hamming scenario with rate $1-t/n - o(1)$ and shared randomness $O(n)$. 
    \item[(b)] 
      There exists a universal code in the oblivious scenario with rate $1-t/n - o(1)$  and  shared randomness $O(\log n)$. 
  \end{itemize}
\end{theorem}

One would expect the rate of a universal code to be lower than the rate of a code that is optimal  for a specific channel.  However,  the $1 - t/n - o(1)$ upper bound is valid for  every  channel  in a large class of channels in the Hamming scenario, defined by graphs in which the left degrees are not much smaller than the maximum right degree (see Remark~\ref{r:rateub}). This class  includes all channels  in the oblivious scenario.  Therefore, surprisingly, the universal codes in Theorem~\ref{t:maininformal}   have  (asymptotically) optimal  rate  even among  codes that are specifically tailored for  each channel satisfying the above condition on left degrees. 

For both codes in Theorem~\ref{t:maininformal}, the universal encoding function is polynomial-time computable, but decoding depends on the channel and requires exponential time. In some settings we obtain polynomial time decoding.
First, in the oblivious scenario, if $t=O(\log n)$ the corresponding decoding functions run in polynomial time. 
Secondly, using concatenation schemes, we obtain  universal codes with efficient encoding and decoding  for weaker classes of channels in which the channel acts on  short blocks of the codeword, rather than on the entire codeword.  Furthermore, in a relaxed version of the  oblivious scenario, in which the distortion vector is chosen at random and independenly on blocks of the codeword, no shared randomness is required: there exists a universal code with encoding and  decoding functions computable by deterministic polynomial-time algorithms. These results together with the full details of the corresponding models are presented in Section~\ref{s:justesen} and Section~\ref{s:piecewise}.


We prove lower bounds for the amount of shared randomness in both scenarios. When  $t$ is a constant fraction of $n$, which is typical in most applications,  the amount of shared randomness is optimal, among universal codes with optimal rate, according to our precise model for shared randomness.\footnote{
  In this model we assume that all randomness is shared, thus no non-shared randomness is used. We are currently investigating a model that allows the encoder  to use both shared and nonshared randomness. Our results indicate that the codes presented here, also use an optimal amount of shared randomness in this more general model. However, the analysis is more difficult and a trade-off between rate and shared randomness exists. The analysis will be given in an upcoming extended version of this paper. An explicit code without shared randomness and non-optimal rate is given in appendix~\ref{s:nonshared}.
  }
Thus, for $t = \Omega(n)$, the universal codes in Theorem~\ref{t:maininformal} are optimal for both rate and randomness.

Theorem~\ref{t:maininformal} (b) shows that the oblivious scenario is a way to restrict channels to allow universal codes with optimal rate and logarithmic shared randomness.     Another sensible way to restrict channels is to bound the computational power of the channel. Under a common hardness assumption which implies the existence of an appropriate type of pseudo-random generators, we show  that there is a universal code with optimal rate, that uses only $O(\log n)$ randomness and is resilient to  all channels that distort adversarially like  in the Hamming scenario except that they use space bounded by a fixed polynomial (for the exact statement, see Theorem~\ref{t:spacebounded}).

 Note that one can always  remove the shared randomness by letting the decoder try all possible random strings. In this way we obtain a list decodable code in which encoding is still probabilistic but decoding is deterministic and with list size exponential in the randomness of the code (the list has one element for each possible random string). Thus, Theorem~\ref{t:maininformal}  (b) implies a universal list decodable code for the oblivious scenario with a deterministic decoder that produces a list of polynomial size, which, with high probability,  contains the message that was encoded.  The same implication can be derived from Theorem~\ref{t:spacebounded}   for channels that compute in bounded space, as described above.

In general, by simple random coding  one can  easily obtain private codes, but this method uses 
many  random bits. In the proof of Theorem~\ref{t:maininformal} (a) the number of shared random bits is reduced by standard pairwise independent hashing.  The proof of Theorem~\ref{t:maininformal} (b) is more involved and uses some recently established properties of condensers related to bipartite matching.

We next present the full details of our model and state the main results formally.

\subsection{Definitions and results}\label{s:defs}

\begin{itemize}[leftmargin=0pt,itemindent=*,label={--}]
  \item 
  A {\em Hamming channel} from a set $\mcX$ to $\mctX$ is a bipartite graph with left set $\mcX$ and right set $\mctX$. 
  The set $\mcX$ represents the set of codewords that are the input of the channel, and $\mctX$ the 
  distorted outputs returned by the channel. On input $x \in \mcX$ the channel may output $\tx \in \mctX$ if $(x, \tx)$ is an edge of the graph.
    The distortion $T$ of the channel is the maximal right degree. 
    We assume that the left degree of each node is at least 1.
   
  \item 
    Let $\mcX$ be an additive group.
  An {\em oblivious channel} is a subset $E$ of $\mcX$. On input a codeword from $\mcX$, the channel adds a codeword from~$E$.
  The distortion $T$ is the size of $E$.
\end{itemize}

\noindent
{\em Example.} 
Consider a bit flip channel that has $n$-bit strings as input and output, and may flip at most $k$ bits. 
This channel can be represented as a Hamming channel. 
Indeed, we have $\mcX = \mctX = \{0,1\}^n$ and a left node is connected to a right node if its Hamming distance is at most~$k$. 
The distortion $T$ of the channel is equal to the size of a Hamming ball of radius~$k$.
The bit flip channel can also be viewed as an oblivious channel. The sum of two bitstrings is defined by bitwise addition modulo 2, 
and the set $E$ contains all strings of Hamming weight at most~$k$.

  \medskip
An encoding function is a mapping $\enc \colon \mcM \times \mcR \rightarrow \mcX$, where the second argument is used for the shared randomness.
A decoding function is a mapping $\dec \colon \mctX \times \mcR \rightarrow \mcM$. We use the notation~$\enc_\rho(x) = \enc(x,\rho)$ and~$\dec_\rho(x) = \dec(x,\rho)$.
  A {\em channel function} $\Ch$ of a Hamming channel is a mapping from left nodes to right nodes.

\begin{figure}[ht]
\centering
\begin{tikzpicture}[->,>=stealth',shorten >=1pt,scale=0.8,auto,node distance=3cm, transform shape]
 \node[] (1) at (1,4) {$\rho$};
  \node[] (2) at (1,2) {$m$};
  \node[fNode](4) at (3,2) {$\enc$};
  \node[rNode](5) at (5,2) {$\cha$};
  \node[fNode](6) at (7.5,2){$\dec$};
  \node(7)[] at (9.5,2) {$m$};
 \path[every node/.style={font=\sffamily\small}]
    (1) edge node  [left]{}(4)
    (2)    edge node [left]{}(4)
(4) edge [] node [above]{$x$} (5)
      (5) edge   node [above]{$\tx$} (6)
        (6) edge (7);
  \draw (1) -| (6);
\end{tikzpicture}
  \caption{Hamming scenario}
\label{f:hamming}
\end{figure}
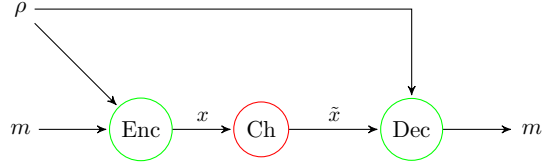


\newcommand{\defReslienceHamming}{
  A private code $\enc \colon \mcM \times \mcR \rightarrow \mcX$ is {\em $(t,\epsilon)$-resilient in the Hamming scenario} 
  if for every $\mctX$ and every Hamming channel from $\mcX$ to $\mctX$  with distortion at most $2^t$, 
  there exists a decoding function $\dec \colon \mctX \times \mcR \rightarrow \mcM$ such that
  for all channel functions $\Ch$ of this channel and all $m \in \mcM$ 
  \[
    \Pr_{\rho \in \mcR}  \left[ \dec_\rho(\Ch(\enc_\rho(m))) = m \right] \;\, \ge\;\, 1-\epsilon.
    \]
  }

  \begin{samepage}
\begin{definition}\label{def:resilientHamming} 
  \defReslienceHamming
\end{definition}
  \end{samepage}
	
\begin{figure}[ht]
\centering
\begin{tikzpicture}[->,>=stealth',shorten >=1pt,scale=0.8,auto,node distance=3cm, transform shape]
 \node[] (1) at (1,4) {$\rho$};
  \node[] (2) at (1,2) {$m$};
  \node[rNode] (3) at (3, 0) {$\cha$};
  \node[fNode](4) at (3,2) {$\enc$};
  \node[nNode](5) at (5,2) {$+$};
  \node[fNode](6) at (7.5,2){$\dec$};
  \node(7)[] at (9.5,2) {$m$};
 \path[every node/.style={font=\sffamily\small}]
       (1) edge node  [left]{}(4)
     (2)    edge node [left]{}(3)
(2)    edge node [left]{}(4)
    (3) edge [] node [below]{$e$} (5)
    (4) edge [] node [above]{$x$} (5)
      (5) edge   node [above]{$\tx=x+e$} (6)
        (6) edge (7);
  \draw (1) -| (6);
\end{tikzpicture}
  \caption{Oblivious scenario}
\label{f:shannon}
\end{figure}
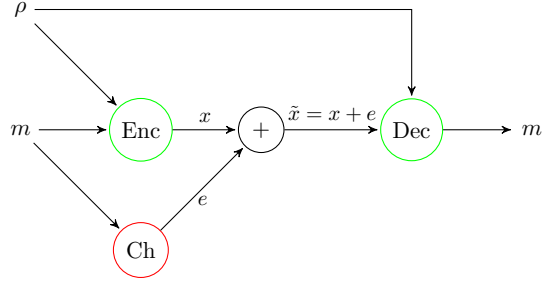

\begin{definition}\label{def:resilientOblivious} 
  Let $\mcX$ be an additive group.
  A private code $\enc \colon \mcM \times \mcR \rightarrow \mcX$ is  {\em $(t,\epsilon)$-resilient in the oblivious scenario} 
  if for every oblivious channel $E \subseteq \mcX$ of size at most $2^t$, there exists a decoding function 
  $\dec \colon \mcX \times \mcR \rightarrow \mcM$ such that for all $m \in \mcM$ and
  $e \in E$
  \begin{equation}\label{eq:oblivious}
     \Pr_{\rho \in \mcR}  \left[ \dec_\rho(\enc_\rho(m)+e) = m \right] \;\, \ge\;\, 1-\epsilon.
  \end{equation}
\end{definition}

We can limit resilience to  a class of channels ${\cal C}$ by replacing in the above definitions ``every Hamming (oblivious) channel'' by ``every Hamming (oblivious) channel in ${\cal C}$''. 
In Section~\ref{s:weakhamming} we analyze an intermediate model, called the  \emph{additive Hamming scenario}, in which the universal code belongs to an additive group and is resilient to all Hamming channels that add an error vector from a fixed set~$E$. 

\bigskip
\noindent
The next two theorems restate the two claims in Theorem~\ref{t:maininformal} with full specification of parameters.

\newcommand{\thmBruno}{
For every $n, t$ and $\epsilon > 0$, 
  there exists a polynomial time computable private code $\enc: \zo^k \times \zo^d \mapping \zo^n$ that is $(t, \epsilon)$-resilient  in the Hamming scenario such that
\begin{itemize}
\item $k \geq n - t - \lceil \log \tfrac 1 \epsilon \rceil$ ,
\item The encoder $\enc$ and the decoder functions $\dec$  share $d = 2n$ random bits.
\end{itemize}
}

\begin{theorem}\label{t:bruno}
  \thmBruno
\end{theorem}

For the results regarding the oblivious scenario, we view $\zo^n$ as the vector space $(\F_2)^n$ in the natural way. 
\begin{theorem}\label{t:main}
There exist constants $c, c'$  such that for every $n,  \tm, \epsilon > 0$,  there exists a polynomial-time computable private code $\enc: \zo^k \times \zo^d \mapping \zo^n$ that is $(\tm, \epsilon)$-resilient  in the oblivious scenario such that
\begin{itemize}
\item $k \geq n - \tm - c\big(\tfrac{\tm}{\log n} \cdot \log (1/\epsilon) + \log (n/\epsilon)\big)$,  
\item The encoder $\enc$ and the decoder functions $\dec$  share $d \le c'(\log n + \log (1/\epsilon))$ random bits.
\end{itemize}
 \end{theorem}

\noindent
Note that if $\log(1/\epsilon) = o(\log n)$, then the rate of the code is $k/n  \geq 1 - \tm/n - o(1)$.

The next code for the oblivious scenario has even better rate (for $\tm$ larger than $\log^4 n$)  but uses more shared random bits.

\begin{theorem}\label{t:codepolylog}
There exist constants $c, c'$  such that for every $n, \tm,  \epsilon > 0$, there exists a polynomial-time computable private code $\enc: \zo^k \times \zo^d \mapping \zo^n$ that is $(\tm, \epsilon)$-resilient  in the oblivious scenario such that
\begin{itemize}
\item $k \geq n - \tm - c(\log^3(n/\epsilon))$,
\item The encoder $\enc$ and the decoder functions $\dec$  share $d \le c'(\log^3 (n /\epsilon))$ random bits.
\end{itemize}
\end{theorem}

\begin{remark}\label{r:complexity}
  \textit{Given $E$ as a list, the decoding algorithms in Theorems~\ref{t:main} and \ref{t:codepolylog} run in time polynomial in the size of $E$ and space polynomial in~$n$.
  Similarly, in Theorem~\ref{t:bruno}, the decoding algorithm runs in time polynomial in the time needed to enumerate all left neighbors of a right node of the channel. Given oracle access to the channel, 
  it runs in space polynomial in~$n$.
  }
\end{remark}

\medskip

\begin{remark}\label{r:linear}
  The codes of Theorems~\ref{t:main} and~\ref{t:codepolylog} are \emph{linear} for fixed randomness, i.e., $\enc(x+y, \rho) = \enc(x, \rho) + \enc(y, \rho)$ over $\F_2$. 
  The code in Theorem~\ref{t:bruno} is affine. 
\end{remark}

\if01
The next code for the oblivious scenario has even better rate (for $\tm$ larger than $\log^3 n$)  but uses more shared random bits.

\begin{theorem}\label{t:codepolylog}\marius{[changed the statement because of the better invertible function]}
There exist constants $c, c'$   such that for every $n, \epsilon > 0, \tm$ satisfying $n - \tm - c(\log \tm  \cdot \log (n/\epsilon)) > 0$, there exists a polynomial-time computable private code $\enc: \zo^k \times \zo^d \mapping \zo^n$ that is $(\tm, \epsilon)$-resilient  in the oblivious scenario such that
\begin{itemize}
\item $k \geq n - \tm - c(\log \tm \cdot \log (n/\epsilon))$,
\item The encoder $\enc$ and the decoder functions $\dec$  share $d \le c'(\log \tm \cdot \log (n /\epsilon))$ random bits.
\end{itemize}
For every oblivious channel $E$, there exists a  corresponding decoding function running  in time polynomial in the time it takes to enumerate $E$.
  Moreover, given oracle access to $E$, there exists a corresponding decoding function running  in space  polynomial in $n$.
\end{theorem}
\fi

\subsection{Related works and comparison with our results}\label{s:related}

The setting of our results  has  two distinctive features: there is no restriction on the type of channel distortion, and the codes we construct  are universal, meaning that the encoder does not know the type of channel he has to cope with.

Channels with general distortion capabilities have been studied starting with the paper of Shannon~\cite{sha:j:communication} that has initiated Information Theory, and which contains  one of the most basic results of this theory, the Channel Coding theorem. In~\cite{sha:j:communication}, a channel is given by probability mass functions $p(y \mid x)$ (one such function for each symbol $x$ in a given finite alphabet), with the interpretation that when $x$ is transmitted, $y$ (also a symbol from a finite alphabet)  is received with probability $p(y \mid x)$. In Shannon's paper, the channel is memoryless: when the  $n$-symbol string $x_1 x_2 \ldots x_n$ is  transmitted, the string  $y_1 y_2 \ldots y_n$ is received with probability $\prod_{i=1}^n  p(y_i \mid x_i)$.   The Channel Coding theorem determines the maximum encoding rate for which decoding is possible with error probability converging to $0$ as $n$ grows. Csisz\'{a}r and  K\"{o}rner~\cite[Theorem 10.8]{csi-kor:b:inftheory} show with a non-explicit construction the existence of an encoder that does not know the memoryless channel.  Verd\`{u} and Han~\cite{ver-han:j:channel} prove a Channel Coding theorem for channels that are not required to be  memoryless (in their model $p(y \mid x)$ is defined for $x$  and $y$ being blocks of $n$ symbols).  We note that to achieve maximum rate, the encoding function in~\cite{ver-han:j:channel} knows the values $p(y \mid x)$, and therefore  it is not universal. 

General channels have also been studied in Zero-Error Information Theory, a subfield in which the goal is that encoding/decoding  have to succeed for \emph{all} transmitted messages.  A channel is given by the set of pairs $S = \{(x,y) \mid p(y \mid x) > 0\}$. $S$ can be viewed as the set of edges of a bipartite graph, with the same interpretation as in our definition for the Hamming scenario: when a left node $x$ is transmitted, the receiver gets one of $x$'s neighbors, chosen  by the channel. One can retain just the graph (ignoring the conditions $p(y \mid x) > 0$, so that the channel behaves adversarially), and obtain a pure combinatorial framework.  Two left nodes $x_1, x_2$ are \emph{separated} if they have no common neighbor, and encoding amounts essentially to finding a set of strings that are pairwise separated, so that they form the codewords of a code. This model is very general, but most results assume that the bipartite graph has certain properties, see the survey paper~\cite{ko-or:j:zinf}.  To the best of our knowledge, all the results assume that the encoding function knows the bipartite graph, and thus it is not universal. The settings in Zero-Error Information Theory and our study have some similar  features:  besides modeling a channel by a bipartite graph,  both of them do not assume any  stochastic process and, furthermore, both of them require encoding/decoding to succeed for all messages (in our setting the success is with high probability over the shared random bits). 

Guruswami and Smith~\cite{guruswami2016optimal} study channels in the oblivious scenario (they call them oblivious channels or additive channels) and in the Hamming scenario, similar to our definitions, except that the channel may only add noise vectors of Hamming weight at most $t$, while in our setting, we may add noise vectors from an arbitrary but fixed set $E$ (of the same size as a Hamming ball of radius~$t$ and this set is only known to the decoder).
In their setting, the encoder is probabilistic and the decoder is deterministic. They obtain codes in the oblivious scenario with polynomial-time encoding and decoding and optimal rate. In our results, the encoder and the decoder share randomness and the decoder is not efficient, but the codes are universal  and are resilient to a more general type of noise, because the set $E$ of noise vectors may contain vectors of any Hamming weight.

\medskip
The concept of a universal code in the Hamming scenario is directly inspired from the universal compressor in~\cite{bau-zim:t:univcompression}.  There, a decompressor $\ttd$ is a (deterministic) partial function mapping strings to strings.
For a string $x$, the Kolmogorov complexity $C_\ttd(x)$ is the length of a shortest string $p$ such that $\ttd(p) = x$.  The probabilistic compression algorithms have a target length~$\ell$ and a target error probability~$\epsilon$ as extra inputs.
More precisely, a compressor $\ttc$  maps every triple (error probability $\epsilon$, length $\ell$, string $x$) 
to a string $\ttc_{\epsilon,\ell}(x)$ of length~$\ell$, representing the compressed version of~$x$.
Such a compressor is {\em universal with overhead $\Delta$} 
if for every decompressor $\ttd$ there exists another decompressor $\ttd'$ such that 
for all triples $(\epsilon, \ell, x)$ with $\ell \ge C_\ttd(x) + \Delta$, we have $\ttd'(\ttc_{\epsilon,\ell}(x)) = x$ with probability~$1-\epsilon$.

It is shown in~\cite{bau-zim:t:univcompression}, that there exists a universal compressor computable in polynomial time  and having polylogarithmic overhead $\Delta$. In other words, for every compressor/decompressor pair $(\ttc, \ttd)$, no matter how slow $\ttc$ is, or even if $\ttc$ is not computable, the universal compressor produces in polynomial time codes that are  almost as short as those of $\ttc$ (the difference in length is the polylogarithmic overhead). The cost is that decompression from such codes is slower. 

The universal compressor also provides an optimal solution to the so-called \emph{document exchange problem}.\footnote{
  This problem is also called \emph{information reconciliation}. 
  In the Information Theory literature it is typically called  \emph{compression with side information at the receiver} or \emph{asymmetric Slepian-Wolf coding}.
} In this problem, Alice holds $x$, the updated version of a file, and Bob holds $y$, an obsolete version of the file. 
Using the universal compressor, Alice can compute in polynomial time a string $q$ of length $t$ which she sends to Bob, and if $t \geq C_\ttd(x \mid y) +\Delta$ (for some decompressor $\ttd$), then Bob can compute $x$ from $y$ and $q$. What is remarkable is that Alice does not know $y$. Moreover, she does not know $\ttd$.  The connection to our setting comes from the fact that a decompressor $\ttd$ is equivalent to a bipartite graph as in our definitions, and the condition $C_\ttd(x \mid y) < t$ is the same as saying that  $x$ is the left neighbor of the right node $y$, which has degree less than $2^t$.

\medskip
As we have already mentioned, the proof of Theorem~\ref{t:bruno} for the Hamming scenario uses random coding and the well-known technique of pairwise-independent hashing to reduce the number of shared random bits from exponential to linear in $n$.  
 Using a pseudo-random generator and a hardness assumption, we can further reduce this to $O(\log n)$ for channels that are computable within some space bound. More precisely, for each polynomial we obtain a code that is resilient to all channels that can be computed with space at most this polynomial. The hardness assumption is that there exists a set computable in E = DTIME($2^{O(n)}$) that is not solvable in subexponential space.  
	The same or similar  hardness assumptions and derandomization arguments have been used before in coding and compression~\cite{tre-vad:c:psamplextractor, che-sho-wig:c:codes, afps:c:lowdepthwit, vin-zim:j:compression}.

The rest of this section regards the proofs of Theorem~\ref{t:main} and Theorem~\ref{t:codepolylog} for the oblivious scenario, which use more advanced techniques to reduce the number of shared random bits to logarithmic in $n$ (respectively, polylogarithmic).
These proofs  are based on  a  similarity that exists  between the document exchange problem and channel coding. In both problems, the receiver needs to reconstruct $x$ from  $y$, which  is close to $x$ in the sense that $C_\ttd(x \mid y) < t$, or, in this paper, $x$ is one of the at most $2^t$ neighbors of $y$ in the bipartite graph that represents the channel (this holds for the Hamming scenario; in the oblivious scenario, a similar ``closeness" relation exists). The difference is that in the document exchange problem, the receiver holds $y$ before transmission, while in channel coding, $y$ is received via transmission and is the channel-distorted version of $x$.  

The connection between the two problems has been exploited  in several papers  starting with the original proof of the Slepian-Wolf theorem~\cite{sle-wol:j:distribcompression}, which solves the document exchange problem  using codes obtained via the standard technique in the Channel Coding Theorem. Wyner~\cite{wyn:j:shannon} gives an alternative proof using linear error correcting codes and syndromes, and there are other papers that have used this idea~\cite{orl:j:intcomp, geh-dra:j:sw, chu-rom:j:sw}. Our approach is similar but works in the other direction: we take linear codes obtained via the method from~\cite{bau-zim:t:univcompression}  for the document exchange problem and use them for channel coding.  

\medskip
The technique used in~\cite{bau-zim:t:univcompression} is based on condensers and is related to previous solutions for several versions of the document exchange problem which used  a stronger tool, namely  extractors~\cite{bfl:j:boundedkolmogorov,muc:j:condcomp,mus-rom-she:j:muchnik,bmvz:j:shortlist,bau-zim:c:linlist,zim:c:kolmslepianwolf}. 
We remark that all  these previous papers do not require linear codes, which are crucial  for the method in this paper.  

It is common to first obtain non-explicit objects using the probabilistic method and then to attempt explicit constructions. In our case, however,  it is not clear how to show the existence of linear extractors with the probabilistic method.  Instead of extractors, we use condensers, and fortunately, a random linear function is a condenser. Moreover, the explicit condensers obtained by Guruswami, Umans, and Vadhan~\cite{guv:j:extractor}, Ta-Shma and Umans~\cite{tas-uma:c:extractor}, and Raz, Reingold and Vadhan~\cite{rareva:j:extractor} (this one is actually an extractor) happen to be linear. 


There are well known similarities between extractors (or their poor relatives, condensers) and error correcting codes (see the discussion in~\cite[Chapt. 8]{vad:b:pseudorand}). In fact the condensers that we use are based on Reed-Solomon codes and Parvaresh-Vardy codes. Closer to our technique is the work of Cheraghchi~\cite{che:c:codescondenser}, who, like us,  uses linear condensers,  linear extractors and the method of syndrome decoding. He obtains codes for the class of  binary symmetric channels, which are channels  that distort by  adding  a \emph{random} noise vector from a given set.  In contrast, our codes for the oblivious scenario defeat channels that add noise \emph{adversarially}, and, moreover, their rates are close to optimal within an additive $o(1)$ term (see Theorem~\ref{t:maininformal}, (b)). 
\if01
 The central pieces in our proof are  the notion of an \emph{invertible function} and the construction of such functions from condensers. These technical innovations (introduced  by us in~\cite{bau-zim:t:univcompression} and further elaborated in Section~\ref{s:step1} and Section~\ref{s:step2}) have allowed us to find a new aspect in the fertile relation between extractors/condensers  and codes. Invertible functions are equivalent to \emph{lossless expanders}, which have numerous applications. Our construction in Theorem~\ref{t:inv} implies, as far as we know,  the first lossless expander with logarithmic seed   and sublinear overhead.  Moreover,  the underlying graphs of this invertible function and of the one in~\cite{bau-zim:t:univcompression} support \emph{online matching}, a property which previous lossless expanders are not known to have. These aspects of invertible functions are discussed in detail in Section~\ref{s:invexp}. We believe that these results are of independent interest. 
\fi

\section{Universal codes for the oblivious scenario}

In this section we first prove Theorem~\ref{t:main} and Theorem~\ref{t:codepolylog}.  We also show that in case the noise vector is chosen randomly and independently on  blocks of the codeword there is a universal code with polynomial-time deterministic encoding and decoding.

\subsection{Proof overview}

The basic idea of our constructions  is to take the code to be a linear subspace of $(\F_2)^n$ picked at random from a class of subspaces. More precisely,  the codewords  belong to the null space of a random linear function $H$, \ie, $Hx= 0$ for all codewords $x$, where $H$  is chosen at random from a certain set of matrices $\cal H$.  The encoder and the decoder share $H$. The decoder receives the noisy $\tx = x+e$,  and, since $H\tx = H(x+e) = Hx+ He = He$,  he knows $He$, which we view as a random fingerprint of $e$ (also called the syndrome of $e$ in the terminology of linear codes). If ${\cal H}$ has certain properties, this allows him to find $e$, assuming that $e$ is within the tolerated noise level.
The next result implements this idea in a simple  way by taking ${\cal H}$ to consist of all matrices of appropriate size. It  has a short proof and produces a  universal code for the oblivious scenario with close-to-optimal  rate for large $n$. It has the disadvantage that the number of shared random bits is more than linear in $n$.

\begin{proposition}
\label{p:sh1}
 For every $n, \tk, \epsilon >0$ such that $\tk+\log(1/\epsilon) < n$,   there exists a private code that is  $(\tk, \epsilon)$-resilient in the oblivious scenario, with rate $1-\tk/n- \delta_n$, where $\delta_n = \log(1/\epsilon)/n$.  

The encoder and the decoder share $(\tk+\log(1/\epsilon)) n$ random bits.
\end{proposition}

\begin{proof}
The encoder and the decoder share a random \emph{linear} function $H: \zo^n \mapping \zo^{\tk+ \log (1/\epsilon)}$.

\noindent
Since $H$ has  rank at most $\tk+\log(1/\epsilon)$, the null space of $H$ has dimension at least $\kt = n- (\tk+\log(1/\epsilon))$.  The encoder  $\enc$ maps every message $m \in \zo^{\kt}$ into the $m$-th element of the null space of $H$ (for details, see Remark~\ref{r:complexity1}). 

Consider now an oblivious channel $E$ of size at most $2^t$ and a message $m$. Let $x = \enc(m, H)$ be the codeword for $m$, and let $\tx = x+ e$, where $e \in E$ is the noise added by a channel. 
Observe that
\begin{equation}\label{e:syndrome}
H \tx =  H(x+e) =  Hx + He = He.
\end{equation}
The decoder $\dec$ works as follows. On input $\tx$ and $H$,  he first computes $p = H \tx$. He knows that $He = p$ (by~\eqref{e:syndrome}), and  he also knows that $e$ belongs to $E$.  For each $e_1 \in (\F_2)^n$ different from $e$, the probability over $H$ that $He = He_1$ is $\epsilon2^{-\tk}$. By the union bound, with probability $1-\epsilon$, there is only one element $e'$ in $E$ such that $H e' = p$, namely $e$. Consequently, $\dec$ can find $e$ with probability $1-\epsilon$, by doing an  exhaustive search.  Next he finds $x = \tx+e$, and finally from $x$ he finds $m$.  

The rate of the code  is $k/n  = 1 - t/n -\log(1/\epsilon)/n$.
\end{proof}
\begin{remark}\label{r:complexity1}
The encoder function $\enc$ in Proposition~\ref{p:sh1} can be computed in time polynomial in $n$ as follows.  First we compute $k$ independent vectors $v_1, \ldots, v_k$ in the null space of $H$ by finding   $k$ solutions of the equation  $H x = 0$ with  $v_i$  having in the last  $k$  coordinates the values  $(0, \ldots, 0,1,0 \ldots 0)$  (the single $1$ is in position $i$). Next, we form the $k$-by-$n$ matrix $G$ having rows $v_1, \ldots, v_k$ and finally $\enc(m, H) = m G$.

On the other hand, the computation of the decoder function  $\dec$ is slow, because it requires the enumeration of all the elements in $E$.
\end{remark}
\medskip

The codes in Theorem~\ref{t:main} and Theorem~\ref{t:codepolylog}   are constructed using pseudo-randomness tools to reduce the space from which $H$ is selected and consequently reduce the number of shared random bits to logarithmic in $n$ (respectively, polylogarithmic in $n$).
The construction of the codes in these two theorems is done in  two steps:
\smallskip

In \emph{Step 1}, we show that a linear \emph{invertible function} (a concept introduced in~\cite{bau-zim:t:univcompression})  can be converted into a universal private resilient code. Step 1 is presented in Section~\ref{s:step1}.
\smallskip

In \emph{Step 2}, we show how  \emph{condensers} (a type of functions that have been studied in the theory of pseudorandomness)  can be used to construct invertible functions. This step is based on the technique in~\cite{bau-zim:t:univcompression} and is presented in Section~\ref{s:step2}.
\smallskip

Theorem~\ref{t:main} and Theorem~\ref{t:codepolylog} are obtained  by taking  condensers built by Guruswami, Umans, and Vadhan~\cite{guv:j:extractor}, Ta-Shma and Umans~\cite{tas-uma:c:extractor} and Raz, Reingold, and Vadhan~\cite{rareva:j:extractor}, and using Step 2 to obtain invertible functions, followed by Step 1, to obtain the codes. The details are presented in Section~\ref{s:step3}.


\subsection{Construction of private universal  codes in the oblivious scenario from linear invertible functions}\label{s:step1}

A $(t, \epsilon)$-invertible function is a probabilistic function that on input $x$ produces a random fingerprint of $x$. The invertibility property requires that there exists a deterministic algorithm that on input a random fingerprint of $x$ and a list $S$, the ``list of suspects,"  of length at most $2^t$ that contains $x$, with probability $1-\epsilon$ correctly identifies $x$ among the suspects. 
To be useful in the construction of codes, we need the invertible function to be \emph{linear} for any fixed value of randomness. Also, in order to obtain codes with good rates, we want the length of the fingerprint to be $t+\Delta$, for small $\Delta$.

\begin{definition}
\begin{enumerate}
\item A function $F: \zo^n \times \zo^d \mapping \zo^{t+\Delta}$ is $(t, \epsilon)$-invertible if there exists a partial  function $g$ 
  mapping a set $S$ of $n$-bit strings and a $(t+\Delta)$-bit string $y$ into $g_S(y) \in \zo^{n}$ 
  such that for every set $S$ containing at most $2^t$ strings and every $x$ in $S$
\begin{equation}\label{e:inv}
\prob_\rho[g_S(F(x, \rho)) = x] \geq 1- \epsilon.
\end{equation}
\item $F$ is \emph{linear} if for every $\rho \in \zo^d$, the function $F(\cdot, \rho)$ is linear, i.e., for every $x_1, x_2 \in \zo^n$, $F(x_1+x_2, \rho) = F(x_1, \rho) + F(x_2, \rho)$, where we view $x_1$ and $x_2$ as elements of the linear space $(\F_2)^n$, and the output of $F$ as an element of the linear space $(\F_2)^{t+\Delta}$.
\end{enumerate}
\end{definition}

The next  proposition shows that, as announced,  a linear, $(t, \epsilon)$-invertible function can be used to construct a $(\tk, \epsilon)$-resilient private code  in the oblivious scenario (and also  in the additive Hamming scenario discussed in section~\ref{s:weakhamming}). In the oblivious scenario, the encoder and the decoder share the random bits used by the invertible function (in the additive Hamming case, they share more random bits, namely $n +$ the random bits of the invertible function).

\begin{proposition}[Invertible function $\rightarrow$ code in the oblivious scenario]  
\label{p:sh2}
  If there exists a linear $(t,\epsilon)$-invertible function $F :\zo^n \times \zo^{d} \mapping \zo^{t+\Delta}$, 
 then there exists a private code $\enc$  that is  $(\tk, \epsilon)$-resilient in the oblivious scenario, with rate $1- (\tk+\Delta)/n$, and such that the encoder and the decoder share  $d$ random bits.
\end{proposition}

\begin{proof}
Since $F(\cdot, \rho)$ is a linear function, it is given by a $(\tk + \Delta)$-by-$n$ matrix $H_\rho$ with entries in $\F_2$, such that $F(x,\rho) = H_\rho x$ (recall that we view $x$ as an $n$-vector over  $\F_2$). The matrices $H_\rho$ are viewed as parity-check matrices of linear codes. 
\medskip

The encoding and decoding procedures are as follows:
\begin{enumerate}
\item The encoder $\enc$  and the decoder $\dec$  share a random string $\rho \in \zo^d$.
\item $\enc$ on input a message $m$  of length  $n - (\tk + \Delta)$ computes the codeword $x$ of length $n$ as follows:
\begin{enumerate}
\item View $m$ as a positive integer in the natural way (based on the base 2 representation of integers).
\item The codeword $x$ is obtained by picking the $m$-th element in the null space of $H_\rho$ (so $H_\rho x = 0$). Note that the dimension of the null space of $H_\rho$ is at least  $n - (\tk + \Delta)$, because the rank of $H_\rho$ is at most $\tk + \Delta$.  Thus the encoder is well defined.

\end{enumerate}
\item Consider an oblivious channel $E$ of size at most $2^t$. 
\item The decoder $\dec$,  on input $\tx = x+e $, where $e  \in E$ is the noise added by the channel,  attempts to find  $m$  as follows:  \label{decoder}
\begin{enumerate}
\item $\dec$ computes  $p = H_\rho \tx$ (i.e.,  $p$ is the syndrome of $\tx$).
\item Note that 
\[
H_\rho \tx = H_\rho (x+e) = H_\rho x + H_\rho e = H_\rho e.
\]
 Thus $p$ is also the syndrome of $e$, and, consequently,   $F(e,\rho) = p$.

\item 
$\dec$ uses the inverter function $g$ given by~\eqref{e:inv}. It runs $g$ on  input $p = F(e, \rho)$ and list~$E$, 
    and with probability $1-\epsilon$, obtains  $e$. Next,  $x = \tx + e$, and finally from $x$, he finds~$m$.

\end{enumerate}
\end{enumerate}

The rate of the code is
 \[
r = \frac{|m|}{|x|} = \frac{n-(\tk + \Delta)} {n}.
\]
\end{proof}

\begin{remark}\label{r:complexity2}
We make the following observations regarding the complexity of the encoder function  $\enc$ and decoder functions $\dec$ in Proposition~\ref{p:sh2}.
The  invertible function is assumed to be linear and thus $F(x, \rho) = H_\rho x$, for some matrix $H_\rho$. If the mapping   $\rho \mapsto H_\rho$ is computable in time polynomial in $n$, then  $\enc$ is computable in time polynomial in $n$. This can be shown in the same way as in Remark~\ref{r:complexity1}. 

For the  invertible $F$ in Theorem~\ref{t:inv} and the one in the proof of Theorem~\ref{t:codepolylog}, the corresponding  inverters $g$  run in time polynomial in a standard  encoding of $S, y$, and $t$ (the latter written in unary). With such a $g$, a simple  inspection of the description in part~\ref{decoder}, reveals that $\dec$ runs in time polynomial in the time it takes to enumerate $E$.

  If the inverter $g$ of $F$ can be evaluated in polynomial space with oracle access to $S$, then $\dec$ is computable in polynomial space given oracle access to the oblivious channel $E$. This is the case for all invertible functions constructed with explicit condensers, obtained through the method in Corollary 2.13 in~\cite{bau-zim:t:univcompression}, which is also used in this paper (this follows from Remark 3 in~\cite{bau-zim:t:univcompression}). 

  An interesting approach to define channels is to use conditional Kolmogorov complexity. We might consider the set $E$ of all distortion vectors that satisfy $C(e \mid n) < t$, and there exist at most $2^t$ such vectors. The corresponding channel is not computable, but on input $n$ and $t$, the set $E$ can be enumerated. If $F$ is online-invertible, then the decoding algorithm explained above can be used with a simple modification of step 4, (c). 
  Each time an element is enumerated in $E$, we rerun the monotone inverse $g$ with the augmented set~$E$. 
  If one of the runs of $g$ halts with some output, then $\dec$ also halts with the same output.  Note that when $e$ is enumerated in $E$, $g$ on input $E$, $p=F(e, \rho)$ and $t$ returns $e$ with probability $1-\epsilon$. By the monotonicity of $g$, later updates of $E$ can not change a given value of $g$ once it has been generated, and this implies that with probability $1-\eps$, no previous runs of $g$ generated a different output. Thus $\dec$ also returns $e$ with probability $1-\epsilon$.
\end{remark}

\subsection{Construction of invertible functions from condensers}\label{s:step2}

A condenser is a type of function that has been studied in the theory of pseudorandomness, which can be seen as a relaxation of randomness extractors (see~\cite{vad:b:pseudorand}). 
Informally speaking, a condenser maps a random variable that is ``sufficiently random'' and ranges over a large set, to another random variable that is ``sufficiently random'' and ranges over a smaller set.

A random variable has {\em min-entropy $t$} if each value has probability at most $2^{-t}$. The statistical distance between two measures $P$ and $Q$ is $\sup |P(S)-Q(S)|$ for a set~$S$.
 Given a set $B$, we denote $U_B$  to be a random variable that is uniformly distributed on $B$. A condenser uses an additional random  variable, which is uniformly distributed over the set of $d$-bit strings, for some small $d$. 

\begin{definition}
  A function $C: \zo^n \times \zo^d \mapping \zo^{m}$ is a $t \rightarrow_{\epsilon} t'$ condenser, if for every $S \subseteq \zo^n$ of size at least $2^t$, the random variable $X= C(U_S, U_{\zo^d})$ is $\epsilon$-close to a random variable $\widetilde{X}$ that has min-entropy at least $t'$. 
\end{definition}

The quantity $ t + d - t'$ is called the \emph{entropy loss} of the condenser, because the input has min-entropy $t+d$ and the output  is close to having min-entropy $t'$. 
We view $C$ as a bipartite graph $G$ in the usual way: the left nodes are the strings in $\zo^n$, the right nodes are the strings in $\zo^m$ and for each $x \in \zo^n$, $\rho \in \zo^d$ there is an edge $(x, C(x, \rho))$ (thus,   for some $x, y$, there may exist multiple edges $(x,y)$).

Conversely, we sometimes view a bipartite graph $G$ as a function having the left side as the domain and the right side as the range and defined by $G(x, \rho) = y$ if the $\rho$th neighbor of $x$ is $y$ (assuming some fixed ordering of neighbors, and interpreting the binary string $\rho$ as the writing in binary of a positive integer).


The invertible function is constructed by concatenating the outputs of two condensers, namely a condenser of Ta-Shma and Umans~\cite{tas-uma:c:extractor} and a condenser of Guruswami, Umans, and Vadhan~\cite{guv:j:extractor} (the latter with a small modification involving simple hashing - see the proof of Corollary 2.13 in~\cite{bau-zim:t:univcompression}). 

Let $G= (L \cup R, \E)$ be a bipartite graph, where we allow $\E$ to be a multiset (i.e., there may be several edges beween two vertices). We say that $G$ admits $(\ell, r)$ matching up to size $K$ if for any set $S \subseteq L$ of size at most $K$, if it is possible to assign to each vertex $x$ in $S$ a set $A_x$ containing at least $\ell$ of its right neighbors (including in the count multiplicities), such that every vertex in $R$ belongs to at most $r$ sets in the family $\{A_x\}_{x \in S}$. In particular if $G$ admits $(\ell, 1)$ matching, the sets assigned to vertices in $S$ are pairwise disjoint.

We use functions that are condensers for an entire range of min-entropies. More precisely, we use families of functions $\{C_n\}$ indexed by $n$ with parameters $d, \epsilon, m,  \tm$ functions of $n$  (all, except $\epsilon$ being positive integers) and satisfying

\begin{equation}\label{e:conductor}
\begin{array}{l}
\text{$C_n$ has type } C_n:\zo^n \times \zo^d \mapping \zo^m,  \text{ and} \\ \\
\text{$C_n$ is an explicit $t \rightarrow_\epsilon t+d-e$ condenser for all $t \le \tm$ such that $2^t \in \nat$}.
\end{array}
\end{equation}

As usual we drop the subscript $n$ in the notation. Note that $C$ has entropy loss bounded by $e$ for inputs with min-entropy $t$ for all $t \le \tm$ such that $2^t \in \nat$.

\begin{theorem}
  \label{t:condensermatching}
Let $C$ be a condenser as in~\eqref{e:conductor}. Then the corresponding graph admits $( (1-4\epsilon)2^d, 2 \cdot t_{max}  \cdot 2^e)$ matching up to size $2^{t_{max}}$.

\end{theorem}

\begin{theorem}
  \label{t:condensermatch1}
Let $C$ be a condenser as in~\eqref{e:conductor}. Then there exists an explicit  bipartite graph $G' = (L' \cup R', \E')$ with left degree $D' = 2^d \cdot u$,  $L'= \zo^n$, $|R'| = 2^m \cdot u^2$  that admits $((1-5\epsilon) D', 1)$ matching up to size $2^{t_{max}}$,  where $u = O(1/\epsilon \cdot n \cdot t_{max} \cdot 2^e)$.

  If $C$ is linear then $G'$ (viewed as a function) is linear as well.

\end{theorem}


As announced, we use the  following condensers.
\begin{theorem}[\cite{tas-uma:c:extractor}, Theorem 3.2, also Theorem 4.1]\label{t:tu}
For every $n$, $\tm \leq n$, $\epsilon \geq 0$  there exists an explicit function $C_{\tu}: \zo^n \times \zo^d \mapping \zo^m$ such that
\begin{enumerate}
\item For every $t \leq \tm$, $C_{\tu}$ is a $t \rightarrow _{\epsilon} t+d- e_{\tu}$ condenser, with $e_{\tu} = O((\tm/\log n) \cdot \log(1/\epsilon) + \log n)$, 
\item $d = O(\log n)$  and $m \leq \tm$. 
\item $C_{\tu}$ is linear. More precisely, for each $y \in \zo^d$, there is a $m$-by-$n$ matrix $A_y$ with entries in $\F_2$ such that $C_{\tu}(x, y) = A_y x$. Furthermore the mapping $y \mapsto A_y$ is computable in time polynomial in $n$.
\end{enumerate}
\end{theorem}
\begin{theorem}[~\cite{guv:j:extractor}, Theorem 4.3, also Theorem 1.7]\label{t:guv}
For every $n$, $\tm \leq n$, $\epsilon \geq 0$ and constant $\alpha$, there exists an explicit function $C_{\guv} : \zo^n \times \zo^d \mapping \zo^m$ such that
\begin{enumerate}
\item For every $t \leq \tm$, $C_{\guv}$ is a $t \rightarrow _{\epsilon} t+d$ condenser.
\item $d = (1+1/\alpha)(\log n + \log \tm + \log(1/\epsilon)) + O(1)$ and $m \leq (1+\alpha)\tm + 2d$. 
\item $C_{\guv}$ is linear. More precisely, for each $y \in \zo^d$, there is a $m$-by-$n$ matrix $A_y$ with entries in $\F_2$ such that $C_{\guv}(x, y) = A_y x$. Furthermore the mapping $y \mapsto A_y$ is computable in time polynomial in $n$.
\end{enumerate}
\end{theorem}
\begin{remark}
The linearity of $C_{\tu}$ and $C_{\guv}$ are not stated explicitly in~\cite{tas-uma:c:extractor}   and~\cite{guv:j:extractor}. We give some explanations in the appendix, Section~\ref{s:linearguv} and Section~\ref{s:lineartu}.
\end{remark}

\begin{theorem}\label{t:inv}
  For every $\tk \leq n$ and $\epsilon > 0$, there exists a linear $(\tk, \epsilon)$ invertible function $F:\zo^n \times \zo^d \mapping \zo^{\tk + \Delta}$ with $d = O(\log n + \log(1/\epsilon))$ and $\Delta = O(\tfrac{\tk}{\log n} \cdot \log(1/\epsilon) + \log (n/\epsilon))$.   Moreover, the inverter $g$ satisfying~\eqref{e:inv} runs in time polynomial in the length of a standard encoding of $S$ and $F(x,\rho)$. 
\end{theorem}
\begin{proof} The invertible function is obtained by combining the condensers  $C_{\tu}$ and  $C_{\guv}$, with parameters set as follows.
\smallskip

$\bullet$  For the $C_{\tu}$ condenser, we take  the parameters $\epsilon$,  and $\tm$ set to be $\tk$. $C_{\tu}$ uses a random string $\rho_1$ of length  $d_1 = O(\log n)$, the output length is $m_1 \leq  \tk$   and it has entropy loss bounded by $e_{\textrm{TU}} = c_1 ((\tk/\log n ) \cdot \log(1/\epsilon_1) +  \log n)$,  for some constant $c_1$.  The graph $G_{\tu}$ corresponding to $C_{\tu}$ admits  $((1-4\epsilon)D_1, 2\cdot \tk \cdot 2^{e_{\tu}})$ matching up to size $2^{\tk}$, by Theorem~\ref{t:condensermatching}. Let
\[
t_1 = \lceil  e_{\textrm{TU}}   +  \log \tk + 1 \rceil.
\]

$\bullet$ For the $C_{\guv}$ condenser, we take the parameters $\epsilon$, $\tm$ set to  $t_1$ and $\alpha$ set to $1$.  $C_{\guv}$ uses a random string of length $d_2= 2(\log n + \log t_1 + \log(1/\epsilon)) +O(1)$, its output length is $m_2 \leq 2 t_1 + O(\log n/\epsilon)$, and  it has zero entropy loss.  The graph $G'_{\guv}$ corresponding to $C_{\guv}$ according to Theorem~\ref{t:condensermatch1} (which guarantees matching with no sharing) has left degree $2^{d_2'} = 2^{d_2} \cdot u$ and the size of its right side is $2^{m'_2} = |2^{m_2}| \cdot u^2$, where $u  = O(1/\epsilon \cdot n \cdot t_1)$. The graph $G'_{\guv}$ admits $((1-5\epsilon) 2^{d_2'}, 1)$ matching up to size $2^{t_1}$.  We view $G'_{\guv}$ as a function in the standard way.
\medskip

We now define the invertible function $F: \zo^n \times \zo^{d_1+d'_2} \mapping \zo^{m_1 + m'_2}$, by
\[
F(x, (\rho_1, \rho_2)) = C_{\tu} (x, \rho_1) \circ G'_{\guv}(x, \rho_2)
\]

We have used $\circ$ to denote string concatenation. Let us check that $F$ is an invertible function with the parameters  claimed in the statement of the theorem. 

 Consider $S \subseteq \zo^n$ of  size $|S| \leq 2^{\tk}$ and let us fix $x \in S$.  Since $G_{\tu}$  admits  $((1-4\epsilon)2^{d_1}, 2^{t_1})$ matching up to size $2^{\tk}$, there is a function that assigns to every element in $S$ a set of $(1-4\epsilon) 2^{d_1}$ of its right neighbors in $G_{\tu}$, such that no right element is assigned to more than $2^{t_1}$ elements in $S$.  Thus if we take $\rho_1$ random in $\zo^{d_1}$, with probability $1-4 \epsilon$, $C_{\tu} (x, \rho_1)$ has at most $2^{t_1}$ neighbors in $S$, one of them being obviously $x$. Let $S_1$ be the set of neighbors of $C_{\tu} (x, \rho_1)$ in $G_{\tu}$ and let $\mathcal{A}$ be the event that  the size of $S_1$ is bounded by $2^{t_1}$. We have argued that the probability of  
$\mathcal{A}$ is at least $1-4 \epsilon$. 

Since $G'_{\guv}$ admits  $((1-5\epsilon) 2^{d_2'}, 1)$ matching up to size $2^{t_1}$, conditioned on $\mathcal{A}$, there is a function that assigns  to each element in $S_1$ a set containing  $(1-5 \epsilon)2^{d'_2}$ of its neighbors in $G'_{\guv}$, such that all these sets are pairwise disjoint. Thus, conditioned  on $\mathcal{A}$, if $\rho_2$ is picked at random in  $\zo^{d'_2}$, with probability at least $(1-5\epsilon)$, $G'_{\guv}(x, \rho_2)$ has a single neighbor in $S_1$, namely $x$.  Let ${\mathcal B}$ be the event that $G'_{\guv}(x, \rho_2)$ has a single neighbor in $S_1$. We have argued that the probability of  ${\mathcal B}$ conditioned by $\mathcal{A}$ is $(1-5 \epsilon)$. 

Let us condition by $\mathcal{A} \cap \mathcal{B}$, which is an event that has probability at least $1-9\epsilon$.  Under this condition,  $F(x, (\rho_1, \rho_2))$ has with probability $1$ a single neighbor in $S$, namely $x$. Thus by an exhaustive search in $S$, one  can retrieve $x$ from   $F(x, (\rho_1, \rho_2))$.  We conclude that with probability $1-9\epsilon$, one can invert $F(x, (\rho_1, \rho_2)$ and find $x$.

 $F$ is linear because each component is linear and the assertions regarding the sizes of $d$ and $\Delta$ can be checked by inspection.
The proof is concluded after a rescaling of $\epsilon$.
\end{proof}

\subsection{Proofs  of Theorem~\ref{t:main} and Theorem~\ref{t:codepolylog}}\label{s:step3}
The proof of Theorem~\ref{t:main}  follows by plugging the invertible function from Theorem~\ref{t:inv} into Proposition~\ref{p:sh2}.
The assertions from Remark~\ref{r:complexity} regarding the computational  complexity of the encoder function $\enc$  and of the decoder functions $\dec$ follow from Remark~\ref{r:complexity2}.

The proof of Theorem~\ref{t:codepolylog} is similar, except that we use a condenser of Raz, Reingold, and Vadhan~\cite{rareva:j:extractor},  instead of the $C_{\tu}$ condenser from~\cite{tas-uma:c:extractor} and the $C_{\guv}$ condenser from~\cite{guv:j:extractor}.
Note that the condenser of Raz, Reingold, and Vadhan is actually an extractor, but we only use the condenser property (extractors have stronger properties than condensers).

 \begin{proof} of Theorem~\ref{t:codepolylog} (sketch)
We use the following condenser.
\begin{theorem}[Theorem 22, (2) in~\cite{rareva:j:extractor}]
\label{t:rrv}
For every $n$, $\tm \leq n$,  $\epsilon \geq 0$, there exists an explicit  function $C_{\rrv} :\zo^n \times \zo^{d}  \rightarrow \zo^{\tm-\Delta}$,  with the following properties:
\begin{enumerate}
\item For every $t \leq \tm$, $C_{\rrv}$  is a $t \rightarrow_{\epsilon} t-\Delta$ condenser
\item $d = O(\log^3(n) \log^2 (1/\epsilon))$ and $\Delta = O(d)$,
\item $C_{\rrv}$ is linear. More precisely, for each $y \in \zo^d$, there is a $m$-by-$n$ matrix $A_y$ with entries in $\F_2$ such that $C_{\rrv}(x, y) = A_y x$. Furthermore the mapping $y \mapsto A_y$ is computable in time polynomial in $n$.
\end{enumerate}
\end{theorem}
\begin{remark}
The linearity of $C_{\rrv}$ is not stated explictly in~\cite{rareva:j:extractor}. We give some explanations in the appendix, Section~\ref{s:linearrrv}.
\end{remark}
The $C_{\rrv}$ condenser has entropy loss $(t+d) - (t-\Delta) = d+\Delta = O(d)$.  Consider the graph $G'_{\rrv}$ corresponding to  $C_{\rrv}$ according to 
Theorem~\ref{t:condensermatch1}. This graph has left degree $2^{d'} = 2^{O(\log^3 n \cdot \log^2 (1/\epsilon))}$, the size of the right side is $2^{\tm + O(\log^3 n \cdot \log^2 (1/\epsilon))}$, and admits $((1-5 \epsilon) 2^{d'}, 1)$ matching up to size $2^{\tm}$.  We define the invertible function $F$, by $F(x, \rho) = G'_{\rrv}(x, \rho)$ (viewing the graph as a function). We check that $F$ is invertible.  Let $S\subseteq L$ with size at most $2^{\tm}$ and let us fix $x \in S$. By the matching property, if $\rho$ is picked at random in $\zo^{d'}$, then with probability $1-5 \epsilon$, $G'_{\rrv}(x, \rho)$ has a single neighbor in $S$, namely $x$. Therefore, if  we do an exhaustive search in $S$ for a neighbor of $F(x, \rho)$ , we find $x$ with probability $1-5\epsilon$, which shows that $F$ is invertible. $F$ is linear because $G'_{\rrv}$ is linear, and the other parameters follow by simple inspection.
\end{proof}

\subsection{Universal codes  with deterministic polynomial-time encoding and decoding  for channels with random distortion in the oblivious scenario} \label{s:justesen}

The universal codes for the oblivious scenario  in Theorem~\ref{t:main} and Theorem~\ref{t:codepolylog} do not have efficient decoding.  In this section, we follow closely Cheraghchi~\cite{che:c:codescondenser} and use Justesen's concatenation scheme~\cite{jus:j:concatcodes} to turn the code from Proposition~\ref{p:sh1} into a code that does not use any randomness and which has an encoder and a decoder that run in time polynomial in the codeword length.  
The cost is that the code works against channels that distort at random, and are memoryless. This means that the codeword consists of a number of blocks, the channel distorts by choosing for each block a random noise vector (instead of the adversarial choice  in Definition~\ref{def:resilientOblivious}), and the random choices for each  block are independent.

We now present formally the model. For each $h$, $\Sigma^h$ denotes the set of binary strings of length $h$, which we identify in the obvious way with $(\F_2)^h$. We assume that the codewords are tuples of elements in  $(\F_2)^n$, for some natural number $n$.  Recall that an oblivious channel with distortion $T$  is given by a set of noise vectors $E \subseteq (\F_2)^n$ of size at most  $T$.   Let $D$ be the number of blocks. A $D$-memoryless oblivious channel  takes as  input a $D$-tuple $(x_1, \ldots, x_D) $  and outputs $(x_1+e_1, \ldots, x_D +e_D)$, where each $e_i$ is chosen uniformly at random  in $E$, and the $D$ random choices are independent.
The encoding function has the type $\enc: (\F_2^k)^S \mapping (\F_2^n)^D$, for some $k$ and $S$.  The rate of the code is $(S \cdot k) /(D \cdot n) $. The decoding function $\dec$ corresponding to such a  channel  has the type $\dec : (\F_2^n)^D \mapping  (\F_2^k)^S$. 

We say that $\enc$ is a universal code $(t, \epsilon)$ resilient against $D$-memoryless oblivious  distortion, if for every $D$-memoryless oblivious channel with distortion at most $2^t$, there is some decoding function $\dec$, such that for every $m \in (\F_2^k)^{S}$, 
\begin{equation}\label{e:randdistort}
\dec(\enc(m)+e) = m,
\end{equation}
with probability $1-\epsilon$ over  $e = (e_1, \ldots, e_D)$ chosen uniformly at random  in $E^D$.
\begin{theorem}\label{t:randomdistort}
For every constant $\alpha > 0$, every $n$, every $t< n - O(1)$ (with the $O(1)$ constant depending on $\alpha$), there exists $\enc: (\Sigma^k)^S \mapping (\Sigma^n)^D$ a universal code $(t, e^{-\Omega(D)})$ resilient against random $D$-memoryless oblivious distortion, such that:
\begin{enumerate}
\item $D= 2^{n (t + O(1))}$ (with the $O(1)$ constant depending on $\alpha$),

\item $S = \lfloor(1-\alpha) D \rfloor$, $k = n-t - O(1)$ (with the $O(1)$ constant depending on $\alpha$). Consequently,  $\enc$ has  rate $S \cdot k /D \cdot n = (1-\alpha)(1-t/n-o(1))$,

\item The encoder $\enc$ is computable in time $O(n^2 D)$ (so encoding runs in time quasilinear in the bit-length of a codeword).
\item For every $D$-memoryless channel in the oblivious scenario with distortion at most $2^t$,  the corresponding decoding function $\dec$ runs in time $((nD)^2)$ (so decoding runs in time quadratic in the bit-length of a codeword).
\end{enumerate}
\end{theorem}
\begin{remark}\label{r:qualityJustesen}
  How good is the rate of the universal code in  Theorem~\ref{t:randomdistort}? Some channels require codes with rate at most,  essentially,  $1-t/n$.
Indeed, consider a Hamming channel $G=(\mcX, \mctX, E \subseteq \mcX \times \mctX)$ with distortion at most $2^t$ (recall that this means that all \emph{right} degrees are at most $2^t$). Suppose the channel has the property that all (or almost all) \emph{left} degrees are at least $2^t$.  This is the case of any oblivious channel (where all left degrees and all right degrees are equal), and also the case of the binary symmetric channel (BSC), for the appropriate definition of $t$.  
By a sphere-packing  argument similar to the one in the proof of Theorem~\ref{t:rate} (see Remark~\ref{r:rateub}), any code that satisfies ~\eqref{e:randdistort} has rate at most $1-t/n + \log(1/(1-\epsilon))/n$. The universal code in Theorem~\ref{t:randomdistort} is quite good for such a channel because its rate is essentially within a factor $1-\alpha$ of the upperbound, and $\alpha$ is an arbitrarily small constant. However some channels have left degree much smaller than the maximum right degree (for example, deletion channels). For such channels, there may exist specific codes tailored for them with better rate than the universal code in  Theorem~\ref{t:randomdistort}.
\end{remark}
\medskip

\noindent
\emph{Proof sketch of Theorem~\ref{t:randomdistort}.}  For full details, we refer to Section~\ref{s:proofranddistort}. The encoding  is obtained by concatenating an ``outer" code $\enc_{out}$ with $D$ ``inner" codes $\enc_{\rho_1}, \ldots, \enc_{\rho_D}$. The inner codes are obtained from the private code $\enc$ in Proposition~\ref{p:sh1}, by fixing the randomness to every $d$-bits string, i.e., $\enc_{\rho}(\cdot) \stackrel{def.}{=} \enc(\cdot, \rho)$. Each inner code maps a $k$-bit string to an  $n$-bit string.  The outer code works with words over the alphabet $\Sigma^k$, and maps an $S$-symbol word $(m_1, \ldots, m_S)$ into a $D$-symbol word $(c_1, \ldots, c_D)$.
Next, in the concatenation step, each $c_i$ is encoded with the inner code $\enc_{\rho_i}$.
\[
 (m_1, \ldots, m_S)  \xrightarrow[\text{\quad outer code \quad }]{} (c_1, \ldots, c_D)   \xrightarrow[\text{\quad \quad  \quad inner codes \quad \quad \quad }]{} (\enc_{\rho_1}(c_1), \ldots,  \enc_{\rho_D}(c_D))
\]
As the outer code we use Spielman's error correcting code~\cite{spi:j:code}. 
This code works over an arbitrary alphabet, has rate $(1-\alpha)$ (for arbitrary $\alpha > 0$),  can correct from a constant fraction of errors, and has encoder and decoder  running  in time linear in the codeword bit-length. The decoder function  of the concatenated code works in the natural way. First it decodes using the  decoding functions of the inner codes. With high probability, all decoding functions, except a small fraction, will decode correctly. Next, the decoding function of the outer code will fix the constant fraction of errors.



\section{Universal codes for the Hamming scenario}\label{s:hamming}

Recall that a Hamming channel is defined by a bipartite graph where the left set $\mcX$ represent codewords and the right set $\mctX$ represent distorted codewords. Each left degree is at least~1, and the distortion is the maximal right degree. 
Also recall definition~\ref{def:resilientHamming} of resilience in the Hamming setting. 

\begin{definition*}[Restated]
  \defReslienceHamming
\end{definition*}

\bigskip
\noindent
We first prove Theorem~\ref{t:bruno}, 
which states that for all $n,t$ and $\epsilon>0$, there exist $(t,\eps)$-resilient codes with $k \ge n-t-\lceil \log \tfrac 1 \epsilon \rceil$ that use $d = 2n$ bits of shared randomness.  Next we consider two restrictions of the Hamming scenario. In Section~\ref{s:spacehamming}, we restrict the computational power of the  channel to algorithms that use space bounded by a given polynomial, and,  under a hardness assumption, we show that there exists a universal code for this  class of channels that uses $O(\log n)$ randomness. In Section~\ref{s:justesen}, we consider weaker channels that choose the distortion on short blocks of the codeword, rather than on the whole codeword, and we show that  there is a  universal code for this class of channels with  polynomial-time encoding and decoding algorithms.

\subsection{Proof of  Theorem~\ref{t:bruno}.}\label{s:proofbruno}

\begin{samepage}
  \begin{theorem*}[Restated]
\thmBruno
\end{theorem*}
\end{samepage}

\bigskip
\noindent
Let $N=2^n$ and $T=2^t$. 
It is enough to prove the theorem for $\epsilon$ being a power of~$2$.
We identify the set of messages with $[K] = \{1, 2, \ldots, K\}$ where $K = \epsilon N/T$. Thus the result follows for $k = \log K$. 

We first construct a code that uses $Kn$ shared random bits. We split these bits in $K$ strings of length $n$ and denote them by
\begin{equation}\label{e:rho}
\rho= (\rho_1, \ldots, \rho_K)
\end{equation}
Thus each $\rho_i$ is  an $n$-bit string chosen independently at random. We define the encoding function $\enc$ by $\enc_\rho(m) = \rho_m$, for each $m \in [K]$. 

We need to prove that this code is $(t,\eps)$-resilient. 
Consider a channel from $\mcX$ to $\mctX$, and for any $\tx \in \mctX$, let $B_\tx$ be the set of left neighbours of $\tx$ in the bipartite graph. 
The size of $B_\tx$ is at most $T$. 
For a fixed $\tx \in \mctX$, by the union bound, the probability that there exists $m' \in [K]$ such that $\enc_\rho(m') \in B_\tx$ is at most $K \cdot T/N \le \epsilon$.  

Let $m \in [K]$ and let $x = \enc_\rho(m)$.  The string $x$ is independent of the value of $\enc_\rho(m')$, for every $m' \in [K]-\{m\}$, and thus, for every channel and every channel function $\cha$, the value of $\cha(x) = \cha(\enc_\rho(m))$ is also independent of $\enc_\rho(m')$. 
Therefore, the probability that for some $m' \not = m$ we have $\enc_\rho(m') \in B_{\cha(x)}$, is also at most~$\epsilon$.  Consequently, with probability at least $1-\epsilon$, one can recover $m$ from $\cha(x)$ and $\rho$ by exhaustive search. 

We now reduce the number of shared random bits from $Kn$ to $2n$. The observation is that in the above argument we only need that the codewords $\enc_\rho(1), \ldots, \enc_\rho(K)$ are pairwise independent. It is well-known that if we pick at random $a, b$ in the field $\F_{2^n}$, and consider the function $h_{a,b}(x) = ax+b$, the values $h_{a,b}(1), h_{a,b}(2), \ldots, h_{a,b}(N-1)$ are pairwise independent. Therefore we replace in $\rho$ from Equation~\eqref{e:rho} each $\rho_i$  by  $h_{a,b}(i)$, for $i=1, \ldots, K$. Now the encoder and the decoder only need to share $a$ and $b$ and the conclusion follows. 

\subsection{Universal codes for space-bounded Hamming channels}
\label{s:spacehamming}

In Theorem~\ref{t:bruno} the number of random bits has been reduced from exponential to $2n$ by pairwise-independent  hashing. If the graph that defines the channel is computationally bounded then the number of random bits can be further reduced to $O(\log n)$ under a reasonable hardness assumption that implies the existence of a convenient pseudo-random generator. 

We consider channels given by graphs for which the edge relation is computable in SPACE[$n^\ell$], for a fixed constant $\ell$. More precisely, for any $\ell$, a SPACE$[n^\ell]$ computable graph is a family of bipartite graphs $(G_n)$, indexed by $n \in \nat$, where the bipartite graphs  have the form $G_n = (\mcX_n = \zo^n, \mctX_n = \zo^{\tilde{n}(n)}, E_n \subseteq \mcX_n \times \mctX_n)$, $\tilde{n}(n)$ is bounded by a polynomial in $n$, 
 and  such that there exists an  algorithm running in space bounded by $n^\ell$ that on input $(x,y)$ returns $1$ if $(x, y)$ is an edge in $G_{|x|}$, and $0$ if it is not. We say that a family of channels is in SPACE[$n^\ell$], if the corresponding family of graphs is in SPACE[$n^\ell$]. 

We show, conditioned on  a hardness assumption, that,   for every constant $\ell$, there exists a private universal code resilient to all families of  channels in SPACE[$n^\ell$], that has optimal rate, and that uses $O(\log n)$ random bits.
\medskip

\textbf{Hardness asumptions and pseudo-random generators}
\smallskip


We use pseudo-random generators that extend  a seed of length $O(\log n)$ to a string of length $n$ in time polynomial in $n$, and such that the output ``looks" uniformly random to certain predicates $A$ of bounded complexity. Formally, a pseudo-random generator $g: \zo^{c \log n} \mapping \zo^n$ fools a predicate $A$ if, for $S$ the uniform distribution on the domain of $g$, and $U$ the uniform distribution on the range of $g$,
\[
| \prob  [A(g(S))=1] - \prob  [A(U)=1] | < 1/n.
\]
 Klivans and van Melkebeek~\cite{km:j:prgenoracle}, relativizing with oracles  the seminal results of  Impagliazzo and Wigderson~~\cite{imp-wig:c:pbpp} and  Nisan and Wigderson~\cite{nis-wig:j:hard},  have shown  that certain hardness assumptions imply the existence of pseudo-random generators of the type that we need. Let $f:\zo^{*} \mapping \zo$ be some function, $A \subseteq \zo^{*}$ be  a set (viewed also as a predicate via the identification with its characteristic function), and let us denote $C_f^{A}(n)$ the size of a smallest circuit with oracle $A$ gates (besides the standard AND, OR, NO gates) that computes the function $f$ for inputs of length $n$.  We denote $\textrm{ E} = \cup_{c > 0} {\rm DTIME}[2^{cn}]$
 \medskip

\emph{Assumption $H(A)$:}
There exists a function $f$ in $\textrm{E}$ such that, for some $\epsilon > 0$, $C_f^{A}(n) > 2^{\epsilon n}$.
\medskip

Let $\textrm{SIZE}^A [n^k]$ denote the set of circuits  of size at most $n^k$ which have oracle $A$ gates  besides the standard gates AND, OR, NOT. 

\begin{theorem}[Klivans and van Melkebeek~\cite{km:j:prgenoracle}]\label{t:prg} If $H(A)$ is true, then for every $k$, there exists a constant $c$ and a pseudo-random generator $g: \zo^{c \log n} \mapping \zo^n$ that is computable in time polynomial in $n$ and fools every  predicate computable by a circuit in $\textrm{SIZE}^A [n^k]$.
\end{theorem}
In our application the set $A$ will be some PSPACE complete problem, say TQBF. 
For such $A$, one can replace $H(A)$ by the following hardness assumption $H_1$  that is less technical and is still plausible.
 \medskip

\emph{Assumption $H_1$:}
There exists a function $f$ in $\textrm{E}$ which is not computable in space $2^{o(n)}$. 
\medskip

 More explictly, this means that $f$ is in $\textrm{E}$, and for every machine $M$ that computes $f$ there exists a constant $\epsilon > 0$ such that, for all sufficiently large $n$, $M$ requires space at least $2^{\epsilon n}$, on some input of length $n$.
\begin{theorem}[Miltersen~\cite{mil:b:derandsurvey}]
For every $A$ in PSPACE/poly, $H_1$ implies $H(A)$.
\end{theorem}

\begin{theorem}\label{t:spacebounded} Assume $H(\textrm{TQBF})$ is true. Then, for every $\ell$, there exists  $c$ with the following property: 
  For every $n,  t, \epsilon > 0$, there exists a polynomial time computable private code $\enc: \zo^k \times \zo^d \mapping \zo^n$ that is $(t, \epsilon)$-resilient  in the Hamming scenario to every family of channels in SPACE[$n^\ell$] such that
\begin{itemize}
\item $k \geq n - t - (\lceil \log(1/\epsilon) \rceil + \log n+ 1)$ ,
\item The encoder $\enc$ and the decoder functions $\dec$  share $d = c \log n$ random bits.
\end{itemize}
Moreover, given oracle access to a channel in Definition~\ref{def:resilientHamming}, there exists  a corresponding decoding function running in space polynomial in $n$.
\end{theorem}

\begin{proof}
We use the notation from the proof of Theorem~\ref{t:bruno}, and follow the construction from this proof. Let us fix $(G_n) $ a family of graphs that is  SPACE$[n^\ell]$ computable and also fix $n$.  Consider the channel defined by $G_n$.   Recall that for $a$ and $b$ in $\F_{2^n}$, we use the function $h_{a,b} (x) = ax + b$.  Each message $m \in [K]$ is encoded as $h_{a,b}(m)$, where the pair $(a,b)$ is the randomness shared between the encoder and the decoder. Recall that for every message $m$, and for every right neighbor $y$  of $x = h_{a,b}(m)$,  with high probability of $(a,b)$, the encoding of no  message other than $m$ is adjacent to $y$.  This allows decoding from $y$ and justifies the following predicate.

  We say that a pair $(a,b)$ is good for $(m,j)$  (where $j$ is a natural number that is at most $2^{\tilde{n}(n)}$) if
\begin{enumerate}
\item  $x = h_{a,b}(m)$ (viewed as a left node in $G_n$) has less than $j$, neighbors  or 
\item if $y$ is the $j$-th neighbor of $x$, for every $m' \in [K]- \{m\}$, $h_{a,b}(m')$ is not a neighbor of $y$. 
\end{enumerate}

For every $(m, j)$, we define the predicate $A_{m, j}$ which  on input $(a,b)$ returns $1$ if $(a,b)$ is good for  $(m, j)$, and  $0$ if it is not good. There exists an algorithm that on input $(m, j, a, b)$ computes  the predicate $A_{m,j}$ on $(a,b)$ and uses  space bounded by a fixed polynomial in $n$ (which is the same for all graphs in  SPACE$[n^\ell]$). Therefore this algorithm is computable by a circuit in  $\textrm{SIZE}^{\textrm{TQBF}} [n^{\ell'}]$, for some constant $\ell'$ which, again,  is the same for all graphs in  SPACE$[n^\ell]$. The last assertion holds because in the standard proof of the PSPACE completeness of TQBF (for example, see~\cite{sip:b:computationthree}), when we reduce a problem in PSPACE to TQBF, the running time of the reduction (and therefore also the size of the circuit computing it)  depends only on the space complexity of the problem. 

We work under the assumption that $H(\textrm{TQBF})$ is true. Theorem~\ref{t:prg} used for $A = \textrm{TQBF}$ and $k = \ell'$ gives a pseudo-random generator $g$ that fools all the predicates $A_{m,j}$ (for all graphs in SPACE$[n^\ell]$), uses a seed of length $c \log n$ (where $c$ is a constant that depends on $\ell$), and is computable in time polynomial in $n$.

It is shown in Theorem~\ref{t:bruno}, that for all $(m,j)$, for random $(a,b)$, the predicate $A_{m,j}(a,b)$ returns $1$ with probability $1-\epsilon$.  Consequently, if we replace the random $(a,b)$ by $g(s)$ with a random seed $s$, the predicate returns $1$ with probability $1-\epsilon'$ for $\epsilon' = \epsilon +  1/n$.  This implies that the encoding function  obtained  by replacing $(a,b)$ by $g(s)$ in the encoding function in Theorem~\ref{t:bruno} satisfies the conclusion of the theorem.
\end{proof}

\subsection{Universal codes with polynomial-time encoding and decoding for piecewise Hamming channels}
\label{s:piecewise}

The universal code in Theorem~\ref{t:bruno} does not have an efficient decoder. Using a concatenation scheme, we show how to obtain a universal code in the shared randomness model with efficient encoding and decoding for a restricted type of  Hamming channels.

The restriction is that the graph $G$ that defines a Hamming channel (see Definition~\ref{def:resilientHamming}) is required to be the product of several graphs, i.e., $G = G_1 \times G_2 \times \ldots \times  G_D$. This means that the vertices of $G$ are $D$-tuples $(u_1, \ldots, u_D)$, where every $u_i $ is a vertex of  $G_i$, and $((x_1, \ldots, x_D), (y_1, \ldots, y_D))$ is an edge of $G$ if $(x_i, y_i)$ is an edge in $G_i$, for every $i \in \{1, \ldots, D\}$. We call such a channel a  \emph{$D$-piecewise Hamming channel}. We recall that such a channel distorts $(x_i, \ldots, x_D)$ into an adversarially chosen $(y_1, \ldots, y_D)$, where for each $i$, $(x_1, y_i)$ is an edge in $G_i$. A piecewise Hamming channel  has distortion $T$, if all graphs $G_i$ have maximum right degree at most $T$.


A universal code ${\cal E}$ has  type  as given  in Definition~\ref{def:resilientHamming}, and ${\cal E}_\rho(m)$ denotes the result of encoding the message $m$, when randomness $\rho$ is used.  We say that ${\cal E}$ is a universal code $(t, \epsilon)$ resilient against $D$-piecewise Hamming distortion, if for every $D$-piecewise Hamming channel $\cha$ with distortion at most $T=2^t$, there is some decoding function $\dec$, such that for every $m$ in the domain of ${\cal E}$, 
\begin{equation}\label{e:piecewiseH}
\dec(\cha({\cal E}_\rho(m), \rho) = m
\end{equation}
with probability $1-\epsilon$ over  the randomness $\rho$ shared by ${\cal E}$ and $\dec$.

Using a concatenation scheme similar to the one in Theorem~\ref{t:randomdistort}, we build a universal code $\cal{E}$ that is $(t, \epsilon)$-resilient to all $D=2^n$-piecewise Hamming channels, with polynomial-time encoding and decoding, and which has rate $(1-\alpha)(1-t/n - o(1))$, where $\alpha > 0$ is an arbitrarily small constant. 

\begin{theorem}\label{t:piecewiseH}
For every constant $\alpha > 0$, every $n$, every $t< n - O(1)$ (with the $O(1)$ constant depending on $\alpha$), there exists ${\cal{E}}: (\Sigma^k)^S \times (\Sigma^{\Delta}) \mapping (\Sigma^n)^D$ a universal code $(t, e^{-\Omega(D)})$ resilient against random $D$-piecewise Hamming distortion, such that:
\begin{enumerate}
\item $D= 2^{n}$,   and the number of shared random bits is  $\Delta = 2nD$,

\item $S = \lfloor(1-\alpha) D \rfloor$, $k = n-t - O(1)$ (with the $O(1)$ constant depending on $\alpha$). Consequently,  $\cal{E}$ has  rate $(S \cdot k) /(D \cdot n) = (1-\alpha)(1-t/n-o(1))$,

\item The encoder $\cal{E}$ is computable in time $O(n D)$ (so encoding runs in time quasilinear in the bit-length of a codeword).
\item For every $D$-piecewise Hamming channel  with distortion at most $2^t$,  the corresponding decoding function $\dec$ runs in time $((nD)^2)$ (so decoding runs in time quadratic in the bit-length of a codeword).
\end{enumerate}
\end{theorem}

The decoder is polynomial-time efficient because it replaces the  exhaustive search in the space of all possible codewords from Theorem~\ref{t:bruno}, with  searches in each segment that forms the piecewise space of codewords. With a high probability, a small fraction of these ``local" searches return incorrect results, but these few errors are repaired  by the outer code. Since each segment contains $n$-bit strings and we take $D=2^n$ segments, this process requires $D \cdot 2^n = D^2$ steps, which is less than quadratic in the length of the codeword (which is $nD$).  The details are presented in Section~\ref{s:proofranddistort}.


\section{Bounds}\label{s:bounds}
In this section we present two kind of bounds for universal codes:  upper bounds for the rate, and lower bounds for the amount of shared randomness.


\subsection{Upper bounds for the rate  of universal codes}\label{s:ub}

If the encoder and the decoder do not use randomness, an upper bound for the rate can be derived via the following standard sphere-packing argument. 
Consider an oblivious channel defined by a set $E$ of size~$T$.
The maximal number of messages we can send with $N$ codewords is equal to $N/T$, because for any 2 messages $m_1$ and $m_2$, the sets $\enc(m_1) +E$ and $\enc(m_2) + E$ must be disjoint. 
The same holds for the Hamming scenario, because we can view the channel as a bipartite graph, (2 nodes are connected if their difference is in $E$), and the right degree is at most~$T$ as well.
In the next theorem,  we adapt this argument for private codes.

\begin{theorem}\label{t:rate}
Let $\enc: \zo^{\kt} \times \mcR \mapping \zo^n$ be a private code that is  $(t, \epsilon)$-resilient in the oblivious scenario, or in the Hamming scenario. Then 
\[
\frac{\kt}{n} \leq 1 - \frac{\tk}{n} + \frac{1+\log(1/(1-\epsilon))}{n}.
\]
\end{theorem}

\begin{proof}
  We consider the oblivious scenario. For the Hamming scenario, the argument is similar.  
  Let $E$ be a set of size exactly $T$.
  For a random selection of $e \in E$, $m \in \zo^\kt$ and $\rho \in \mcR$, we have
  \[
    \Pr_{e,m,\rho} \left[\dec_\rho(\enc_\rho(m)+e) = m \right] \;\ge\; 1-\epsilon\,.
  \]
  For $\rho \in \mcR$, consider the set 
  \[
    A_\rho = \big\{(m,e) : \dec_\rho(\enc_\rho(m) + e) = m \big\}.
  \]
  For a random $\rho$, we have
  \[
    \mathbb E \left[\frac{\# A_\rho}{2^k \cdot 2^t}\right] \ge 1-\epsilon,
  \]
  because the left-hand side is precisely the probability above.
  This implies that there must exist a $\rho \in \mcR$ for which 
  $
    \# A_\rho \ge (1-\epsilon) 2^{k+t}.
  $
  Fix such a~$\rho$.
  Note that for no two  pairs $(m,e)$ in $A_\rho$, the value of $\enc_\rho(m)+e$ can be equal. 
  Hence, $2^n \ge \# A_\rho$. The statement of the theorem follows by combining these 2 inequalities.
\end{proof}

\begin{remark}\label{r:rateub}
The argument in the above proof establishes a stronger assertion.  Consider a Hamming channel defined by a bipartite graph $G$, and assume that every \emph{left} node has  degree at most $2^t$ (note that the noise level is defined using the  degree of \emph{right} nodes). Let $\enc$ be an encoding that defeats this channel, i.e., $\enc: \zo^{\kt} \times \mcR \mapping \zo^n$, and there exists $\dec$ such that  for every $m \in \zo^{\kt}$, with probability $1-\epsilon$ of $\rho \in \mcR$, for all neighbors $\tx$ of $\enc(m, \rho)$, $\dec(\tx, \rho) = m$. Then, the argument shows the same upper bound for the rate $k/n$, as the one in Theorem~\ref{t:rate}. Thus, the upper bound holds not only for universal codes which have to defeat \emph{all} Hamming channels, but also for codes that defeat any \emph{single} Hamming channel satisfying the above  left degree condition. This class of codes includes all channels in the oblivious scenario, because for such channels  the left degree and the right degrees are equal.
\end{remark}

\subsection{Lower bounds for the randomness of universal codes}\label{s:lb}

We first  note that there exist universal codes in which the encoder is randomized and the decoder is deterministic, and, thus they  \emph{do not share} randomness. We provide a non-explicit construction of such a code in Appendix~\ref{s:nonshared}. This code does not achieve an optimal rate. 
In an upcoming extended version of this paper, we show that for some choices of $k$ in the oblivious scenario, any universal code that is $(\eps,t)$-resilient and has optimal rate must use \emph{shared} randomness. In general the trade-off between shared randomness and rate for universal codes is very intricate and for a (lengthy) discussion we refer to the extended version.

Therefore, in what follows we restrict to private codes, \ie, to the model in which the universal encoder and the channel-dependent decoders share randomness, and the encoder does not have access to other types of randomness.  We show lower bounds for the number of random bits in both the Hamming and oblivious scenarios.\footnote{In Appendix~\ref{s:weakhamming},  we discuss a different model, which is intermediate between oblivious and Hamming.}

 We first show that  for  any private universal code in the oblivious scenario, the encoding function  must use at least $\Omega(\log t)$ random bits, regardless of rate, where $t$ is the noise level. 
The universal code for the oblivious scenario in Theorem~\ref{t:main} has $O(\log n)$  random bits, and has optimal rate in the asymptotical sense. Thus the number of  random bits in Theorem~\ref{t:main} matches the lower bound (up to the constant hidden in the $O(\cdot)$ notation),  in the case of  noise level $t=\Omega(n)$, which is typical. 

\begin{theorem}\label{t:lbobliv}
  If $\#\mcM \ge 2$, $\epsilon < 1/2$ and $\enc\colon\mcM \times \mcR \mapping \mcX$ is a private $(t, \epsilon)$-resilient code in the oblivious scenario, then $\#\mcR > t$, i.e., $\enc$ requires more than $\log t$ random bits.
\end{theorem}

\begin{proof}
  It is enough to prove the theorem for only 2 messages.  Let $\mcM = \{\texttt{a}, \texttt{b}\}$ and $\mcR = \{1,2,\ldots,D\}$.
  Consider the channel defined by the set $E$ given by the span of the vectors
  \[
    v_1 = \enc_1(\texttt{a}) - \enc_1(\texttt{b}), \quad \ldots, \quad v_D = \enc_D(\texttt{a}) - \enc_D(\texttt{b}).
  \]
  Thus, $E$ has size at most~$2^D$.
  We need to select $m \in \mcM$ and $e \in E$ such that the probability in~\eqref{eq:oblivious} is at most~$1/2$.
  In the requirement~\eqref{eq:oblivious}, the only relevant values of $\dec$ are vectors of the form
  \[
     \enc_\rho(\texttt{a}) + c_1v_1 + \cdots + c_Dv_D,
  \]
  with $\rho \in \mcR$ and $c \in \{0,1\}^D$.
  Select $\rho$ and $c$ randomly and consider the value of $\dec_\rho$ on the above vector, which is a value in~$\mcM$.
  Note that if we used message $\texttt{b}$ instead of $\texttt{a}$ in the expression above, 
  then the probabilities with which the messages appear do not change (since this corresponds to flipping all bits of~$c$). 
  Assume that the value $\texttt{b}$ appears with probability at least $1/2$.
  If this is not the case, we flip the roles of $\texttt{a}$ and $\texttt{b}$ in the expression above and the explanations below. 
  There exists a choice of $c \in \{0,1\}^D$ such that for at least half of the values $\rho \in \mcR$, 
  the value of $\dec_\rho$ for the above vector is equal to~$\texttt b$.
  Let $e = c_1v_1 + \cdots + c_Dv_D \in E$ be the corresponding vector. 
  For $m = \texttt{a}$, the probability in~\eqref{eq:oblivious} is at most~$1/2$. 
  Hence, for $\eps< 1/2$ the inequality is false, and this implies that if $D\le t$ 
  equation~\eqref{eq:oblivious} can not be satisfied.
\end{proof}

We prove a similar result for the Hamming scenario.

\begin{theorem}\label{t:hamlb}
  If $\#\mcM \ge 2$, $2^t \leq \# \mcX$, $\epsilon < 1/3$ and $\enc\colon\mcM \times \mcR \mapping \mcX$ is a private $(t, \epsilon)$-resilient code in the Hamming scenario, then $\#\mcR > 2^{2t-2}$, i.e., $\enc$ requires more than $2t-2$ random bits.
\end{theorem}

Again, it is enough to prove the statement for two messages.  Let $\mcM = \{\texttt{a}, \texttt{b}\}$ and $\mcR = \{1,2,\ldots,D\}$. Thus, we are given a universal code $\enc :  \mcM \times [D] \mapping [N]$ for some  arbitrary $N$ and the code is resilient in the Hamming scenario up to distortion $2T$ with probability $\epsilon$, where $T=2^{t-1}$.  This means that for every bipartite graph with $N$ left nodes and $N$ right nodes, with degree of every right node $\leq 2T$,   the event (when $\rho$ is chosen at random in $[D]$)
\begin{equation*}\tag{*}\label{eq:cond}
  \enc(a,\rho) \text{ and  } \enc(b,\rho) \text{ have a common neighbor.}
\end{equation*}
has probability at most $\epsilon$.
We show in the next lemma that if $\eps <   1/3$, then $D  > T^2$, from which the conclusion follows

\begin{lemma}\label{lem:diagonalization}
 For every encoding function $\enc \colon \{a,b\} \times [T^2] \rightarrow [N]$, 
  there exists a bipartite graph of the above type such that the event in~\eqref{eq:cond} has probability at least $\eps \geq  1/3$.
\end{lemma}

\begin{proof}
  We construct a bipartite graph with the set of left nodes and right nodes both equal to $[N]$, and  with left and right degrees at most~$2T$ (thus the lower bound is valid even for channels where the left degree is also bounded by $2T$).
  Consider the matrix obtained by setting the $(x,y)$-th entry equal to the number of 
  random strings $\rho$ for which~$\enc(a,\rho) = x$ and $\enc(b,\rho) = y$.
  Since there are $T^2$ strings $\rho$, the sum of all entries of this matrix is $T^2$ as well.

  The  {\em weight} of a column is the sum of all its entries. Similarly for the weight of a row.  A column is heavy if its weight is $\geq T$ and a  heavy row is defined in the same way.
 Note that there are at most $T$ heavy  rows and at most $T$ heavy columns.
  We consider 3 cases:
  \begin{itemize}[leftmargin=*]
    \item The set of heavy columns  have total weight at least $T^2/3$.
    \item The set of heavy rows    have total weight at least $T^2/3$.
    \item None of the conditions above are true.
  \end{itemize}
  In the last case the construction is easy. We set all entries of heavy columns and rows equal to zero. 
  The remaining matrix has weight at least~$T^2/3$, and all its rows and columns have weight less than~$T$ (because they are not heavy).

 We define the   bipartite graph in which
  a left node $x$ is connected to a right node $y$ if $x=y$ or the $(x,y)$ entry of the matrix is positive.

  Since the matrix contains nonnegative integers, every column has less than $T$ positive entries, and hence 
  every left node has degree at most~$T$. By a symmetric argument with rows, 
  we conclude that also the right degrees are at most~$T$. 

  We prove that  the event \eqref{eq:cond} has probability at least $1/3$. Indeed, select $\rho$ randomly, and let $x = \enc(a,\rho)$ and
  $y = \enc(b,\rho)$. 
  With probability at least $1/3$ the entry $(x,y)$ is positive, and this implies that $x$ is a neighbor of both $x$ and $y$.
  In the last case the lemma is satisfied.

  \medskip
  Note that the first and second case are symmetric after flipping the first and second message in~$\enc$.
  Hence, it remains to prove the claim for the second case. 
  In the matrix, we set all rows that have weight less than~$T$ equal to zero. 
  The assumption states that the remaining matrix has weight at least~$T^2/3$.

  The idea to prove  \eqref{eq:cond}, is to consider a set of $T$ values $y$, which we  call {\em pointers}. 
  We connect each heavy row to every pointer.
  Each nonzero column will be connected to a single pointer as well. 
  Since there are at most $T^2$ nonzero columns, 
  we can indeed satisfy the degree bound using at most $T$ pointers. 
  Finally, choose $\Ch(y)$ to be this pointer for each nonzero column $y$.
  Now the inequality fails for $m=2$, since with probability $1/3$, we have that $\enc(a,\rho)$ is a heavy row and that $\enc(b,\rho)$ is a nonzero column. 
  Hence, they are both connected to the pointer~$\Ch(\enc(b,\rho))$. Now the details.

  By the assumption $N \ge 2T$  and taking into account that there are at most $T$ heavy  rows, we can select $T$ rows containing only zeros. The $T$ choosen rows
  are called {\em pointers}.
  We assign  to each   nonzero column  a pointer so that no pointer is assigned to more than $T$ columns.
  Note that there are at most $T^2$ nonzero columns and $T$ pointers, and thus this assignment  is possible.

  The bipartite graph connects a left node $x$ to a right node $y$ 
  \begin{itemize}
    \item if $x$ is a heavy row and $y$ is a pointer, or
    \item if $x$ is a nonzero column and $y$ is its associated pointer.
  \end{itemize}
  The conditions on the degree are satisfied, because every left node 
  is only connected to pointers, and there are at most $T$ of them. 
  Every right node $y$ has degree at most $2T$, because we only need to check 
  this for pointers $y$, and they are connected to $T$ heavy rows and to 
  at most $T$ nonzero associated columns.

  Finally, we need to prove that the event  \eqref{eq:cond} has probability at least $1/3$. For each nonzero column  $y$, let $\Ch(y)$ be the associated pointer, and so also a neighbor of  $y$. 
  With probability $1/3$ for a random $\rho$, the value of $\enc(a,\rho)$ will be a heavy row and $\enc(b,\rho)$ a nonzero column. 
  This means that $\Ch(\enc(b,\rho))$ is a pointer, and hence connected to all heavy rows, thus in particular it is also a neighbor of ~$\enc(a,\rho)$.
  Thus, the event in  \eqref{eq:cond}  happens with probability at least~$1/3$.
\end{proof}

\section{Final comments}  

In our main results, Theorem~\ref{t:bruno},  Theorem~\ref{t:main}, and Theorem~\ref{t:codepolylog}, the encoding function is computable in time polynomial in $n$, but the channel-dependent decoding functions are not efficiently computable (except for the special cases that have been mentioned). 
This is to be expected given the strong universality property of the code.  We  have shown that for memoryless oblivious channels, and also for piecewise Hamming channels, there are universal codes with polynomial-time encoding and decoding (Theorem~\ref{t:randomdistort} and Theorem~\ref{t:piecewiseH}). It would be interesting to find other classes of channels that admit efficient universal codes. It seems natural to consider codes that are resilient to channels that \emph{compute} the distortion using algorithms with low computational complexity.  We have in mind channels that are similar to the computational channels proposed by Lipton~\cite{lip:c:codes}, but which are more general because  the distortion is bounded using our general setting for the Hamming scenario or the oblivious scenario, not by the Hamming weight of the error vector as in~\cite{lip:c:codes}.   Obtaining codes that are resilient to all channels with general distortion capabilities,  that are computable by algorithms in a given complexity class (say,  LOGSPACE, or ${\rm NC}^1$, or finite automata)  and that have efficient encoding and decoding would be very interesting even if they have non-optimal rate.


\section*{Acknowledgements} We are grateful to Andrei Romashchenko for  the helpful conversations we had. We also thank Alexander Shen, for suggesting the definition of ``invertible function,''  which turned out to be important for this paper. 
\bibliography{theory-3}

\bibliographystyle{alpha}

\newpage

\appendix

\section{Appendix}

\subsection{Proofs of Theorem~\ref{t:randomdistort} and Theorem~\ref{t:piecewiseH}}\label{s:proofranddistort}

\textbf{Proof of Theorem~\ref{t:randomdistort}.}  The construction uses the Justesen concatenation scheme, which combines an \emph{outer code}, with several \emph{inner codes.} First the message written with symbols from  the alphabet $\F_2^k$ (for some integer $k$)  is mapped by the encoder of the outer code  into a codeword over the same alphabet, and next each symbol of the codeword is mapped using the encoder of one the inner codes.

\emph{The outer code:}  Following Cheraghchi~\cite{che:c:codescondenser}, we use the linear time encodable/decodable  code constructed by Spielman~\cite{spi:j:code} as the outer code.
\begin{theorem}[~\cite{spi:j:code}]
\label{t:spielman}
For every constant $\alpha < 1$ and every positive integer $k$, there exist a constant $\beta_{\text{Spielman}} > 0$ and an explicit family of codes $(C_S)_{S \in \nat}$ over the alphabet $\F_2^k$, such that $C_S$ encodes messages of length $S$,   has rate $1-\alpha$,  and is error correcting for a fraction $\beta_{\text{Spielman}}$  of errors.The encoder and the decoder run in time that is linear in the bit-length of the codewords.  

More explicly,  for every $S \in \nat$, the encoder of $C_S$ maps $(\F_2^k)^S$ into $(\F_2^k)^D$, with $S/D \geq 1 - \alpha$,   and for every codeword $x \in (\F_2^k)^D$, and every $\tx$ such that that the relative Hamming distance between $x$ and $\tx$ is $\beta_{\text{Spielman}}$, the decoder on input $\tx$ returns the message encoded as $x$. The encoder and the decoder run in time $O(kD)$.
\end{theorem}

\emph{The inner codes:} The inner codes are obtained from the universal code from Proposition~\ref{p:sh1}. Recall that the encoder of this code is a function $\enc (m, H)$, where $H$ is a random 
$(t + \log(1/\epsilon)) \times n$ matrix $H$ with elements in $\F_2$. For $k=n - (t+ \log(1/\epsilon)$,  the encoder maps a $k$-bit message $m$ into an $n$-bit codeword,  which is the $m$-th element of the null space of $H$ (using some canonical ordering of the elements in the null space).  The value of $\epsilon$ will be picked later.  Let $d = n (t+\log (1/\epsilon))$. We use, as inner codes,  $D = 2^d$ codes $\enc_{H_1}, \enc_{H_2}, \ldots, \enc_{H_D}$, which are obtained from $\enc$ in Proposition~\ref{p:sh1}, by fixing the randomness to every $d$-bits string, i.e., for each $H \in \zo^d$,  $\enc_{H}(\cdot) \stackrel{def.}{=} \enc(\cdot, H)$. Each inner code maps a $k$-bit string, viewed in the natural way as an element of $\F_2^k$  into an $n$-bit string, which similarly is viewed in the natural way as an element of $\F_2^n$.

Equipped with  the outer code and the inner codes, we construct a new encoder ${\cal E}$  as  follows.

\emph{Concatenation scheme:} First, the $S$-symbol input message  $(m_1, \ldots, m_S)$ is encoded with the outer code into a $D$-symbol word $(c_1, \ldots, c_D)$.
Next, each $c_i$ is encoded with the inner code $\enc_{\rho_i}$.
\[
\begin{array}{l}
 (m_1, \ldots, m_S)  \\

\quad \quad \bigg \downarrow   \text{\quad \quad\quad outer code encoding;  each $m_i$ and $c_i$ is $k$-bits long } \\

 (c_1, \ldots, c_D)   \\

\quad \quad \bigg \downarrow \text{\quad \quad  \quad inner codes encoding } \\

 (\enc_{\rho_1}(c_1), \ldots,  \enc_{\rho_D}(c_D))
\end{array}
\]

Thus ${\cal E}$ maps binary strings of length $k\cdot  S$ into codewords that are binary strings of length  $n \cdot D$.

Consider now a $D$-memoryless channel $\cha$ in the oblivious scenario adjusted def for the oblivious scenario  that has distortion bounded by $t$. 
Recall that this means that: 
\begin{enumerate}
\item  there is  a set $E \subseteq \F_2^n$ of size $T = 2^t$, and 
\item  $\cha$ takes as  input a $D$-tuple $(x_1, \ldots, x_D) \in (\F_2^n)^D$ , and outputs $(y_1, \ldots, y_D)$, where each $y_i = x_i+e_i$,  for $e_i$  chosen uniformly at random in $E$,  independently of the other choices.  
\end{enumerate}

Suppose the sender encodes the message $(m_1, \ldots, m_S)$ into  the codeword $(x_1, \ldots, x_D)$ (where each $x_i \in \F_2^n$) and the channel distorts it into $(y_1, \ldots, y_D)$. 
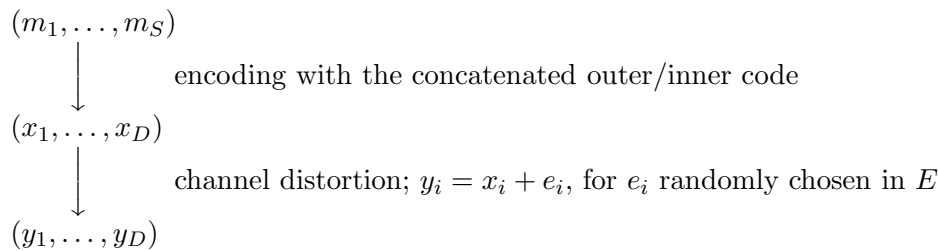
\[
\begin{array}{l}
(m_1, \ldots, m_S) \\

\quad \quad \bigg \downarrow   \text{\quad \quad\quad  encoding with the concatenated  outer/inner code } \\

 (x_1, \ldots, x_D)  \\

\quad \quad \bigg \downarrow   \text{\quad \quad\quad channel distortion; $y_i=x_i + e_i$, for $e_i$ randomly chosen in $E$ } \\

 (y_1, \ldots, y_D)   \\

\end{array}
\]

Henceforth, we consider that $(x_1, \ldots, x_D)$ is fixed, and $(y_1, \ldots, y_D)$ is a random variable depending on the randomness of the channel. The decoder of the concatenated code, first calls the decoders of the $D$ inner codes respectively on each component of $(y_1, \ldots, y_D)$, which return the word $(z_1, \ldots, z_D)$, and next calls the decoder of the outer code on this latter word. We show that with high probability, this procedure reconstructs $(m_1, \ldots, m_S)$.

For each $e \in E$, we say that the matrix $H$ is \emph{good for $e$}  if $H e_1 \not= He$ for all $e_1 \in E - \{e\}$. In the proof of Proposition~\ref{p:sh1},   it is shown that if $H$ is good for $e$, then the decoder on input $x+e$ reconstructs $x$ and that for each $e \in E$, at least a fraction of $(1-\epsilon)$ of $H$'s are good for $e$.  By a standard averaging argument, it follows that  there is a set (of ``good" matrices) GOOD$_{\text{rand}}$ containing $(1-\sqrt{\epsilon})$ fraction of $H$'s  in $\zo^d$ and  a set (of ``good" noise vectors)  GOOD$_{\text{noise}}$  containing $(1-\sqrt{\epsilon})$   fraction of $e$'s  in $E$, such that every $H$ in GOOD$_{\text{rand}}$ is good for every $e$ in  GOOD$_{\text{noise}}$. By rearranging the tuple $(y_1, \ldots, y_D)$, we can assume that the first $(1-\sqrt{\epsilon})D$ components correspond to the ``good" matrices.  On every  $i$ in this segment of ``good" components,  if  $e_i =y_i - x_i$  is a ``good" noise vector,  the inner decoder on input $y_i$ correctly reconstructs $x_i$. Note that the probability (over the randomness of the channel)  that $e_i$ is a ``good" noise vector  is at least $1 - \sqrt{\epsilon}$.

Let $\mu$ be  the expected number of components of $(y_1, \ldots, y_D)$ on which the inner decoders are incorrect. Note that 
\[
\mu \le \sqrt{\epsilon} \cdot D + (1-\sqrt{\epsilon}) \cdot D \cdot \sqrt{\epsilon}.  
\]
In the sum above, the first term  corresponds to the errors made by the inner decoders in the  bad segment and the second term corresponds to the errors made by decoders in the good segment  for which the channel is choosing a bad $e$.  Thus, $\mu < \mu_H$, where $\mu_H \stackrel{def.}{=} 2 \sqrt{\epsilon}\cdot D$. We take $\gamma = \sqrt{\epsilon} \cdot D$.
\[
\prob [ \# \text{ errors } \ge 3 \sqrt{\epsilon} \cdot D] = \prob [ \# \text{ errors } \ge \mu_H + \gamma] \le e^{-2 \gamma^2/D} = e^{-2\epsilon \cdot D}. 
\]
 We have used a form of the Chernoff bounds~\footnote{Let $X = \sum_{i=1}^D X_i$, where $X_i$, $i \in \{1,D\}$ are independently distributed in $[0,1]$. Let $\mu_H \geq \mu$, where $\mu$ is the expected value of $X$. Then for every $\gamma > 0$, $\prob[X \geq \mu_H + \gamma] < e^{-2 \gamma^2/D}$.} for the case when we know an upperbound of the expected value (see Exercise 1.1 (a) in~\cite{dub-pan:b:concentration})

Take $\epsilon$ such that $3 \sqrt{\epsilon} \le \beta_{\text{Spielman}}$. Then, with probability at least $1 -  e^{-2\epsilon \cdot D}$, the inner decoders are correct on all except at most  
$\beta_{\text{Spielman}}$ fraction of positions. In such a case, the outer decoder is able to reconstruct the codeword $(x_1, \ldots, x_D)$ and then the message that is encoded into this codeword.

We now evaluate the runtime of the encoder and the decoder. The encoder calls first the encoder of the outer code, which takes time $O(nD)$, and then calls the encoders of the $D$ inner codes, each one running $O(n^2)$ steps (by Remark~\ref{r:complexity1}). Thus, the total time for encoding is $O(n^2 D)$, which is quasi-linear in $nD$ (the bit-length of a codeword).  The decoder first calls the $D$ inner deccoders, and each one runs in time $O(T n^2)$, so this step takes $O(DTn^2)$, which is $O((nD)^2)$ because $T = O(D)$. Next, it calls the decoder of the outer code, which runs in time $O(nD)$. So the total time is less than $O((Dn)^2)$, so quadratic in the bit-length of a codeword.

\bigskip

\textbf{Proof   of Theorem~\ref{t:piecewiseH}.}  We use again a concatenation scheme with an outer code and an innner code. The outer code is Spielman's error correcting code~\cite{spi:j:code}, presented in Theorem~\ref{t:spielman}.  The inner code is the code from Theorem~\ref{t:bruno}. Its encoder
$\enc_\rho(m)$ maps a $k$-bit message $m$ into an $n$-bit codeword, and $\rho$ represents the  random string used for both encoding and decoding (we recall that we are in the shared randomness setting). The outer code works with strings having symbols from the alphabet $\Sigma^k$.  The encoder of ${\cal E}$ first calls the encoder of the outer code, which  maps an $S$-symbol message $(m_1, \ldots, m_S)$ into a $D$-symbol codeword $(c_1, \ldots, c_D)$.  Next, each $c_i$ is encoded with the inner code $\enc_{\rho_i}$. The encoder and the decoder of ${\cal E}$ share randomness $\rho=(\rho_1, \ldots, \rho_D)$. 

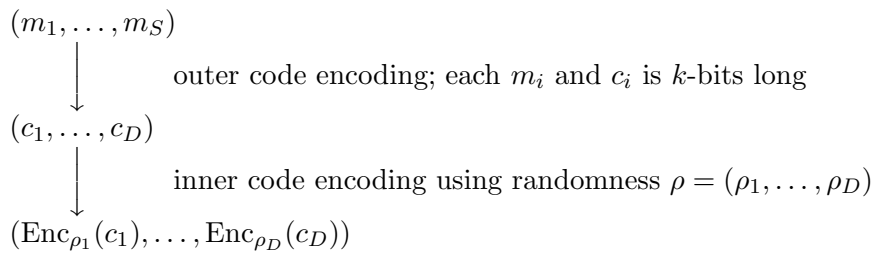
\[
\begin{array}{l}
 (m_1, \ldots, m_S)  \\

\quad \quad \bigg \downarrow   \text{\quad \quad\quad outer code encoding;  each $m_i$ and $c_i$ is $k$-bits long } \\

 (c_1, \ldots, c_D)   \\

\quad \quad \bigg \downarrow \text{\quad \quad  \quad inner code encoding using randomness $\rho = (\rho_1, \ldots, \rho_D)$ } \\

 (\enc_{\rho_1}(c_1), \ldots,  \enc_{\rho_D}(c_D))
\end{array}
\]


Thus, the encoder of ${\cal E}$ maps a $kS$-bit string into an $nD$-bit string, and consequently has rate $(kS) /(nD) = (S/D) \cdot (k/n) \ge (1-\alpha) \cdot (1-t/n - o(1))$.

Suppose a piecewise Hamming channel  with graph $G = G_1 \times \ldots \times G_D$ distorts $(x_1, \ldots, x_D)$ into $(y_1, \ldots, y_D)$.
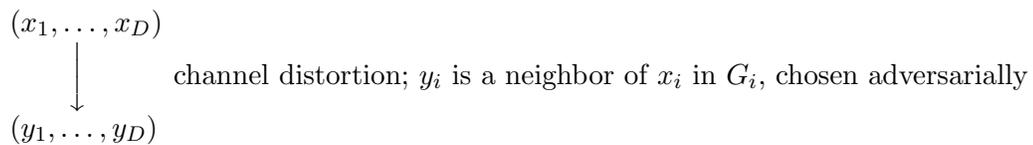
\[
\begin{array}{l} 
 (x_1, \ldots, x_D)  \\

\quad \quad \bigg \downarrow   \text{\quad \quad\quad channel distortion; $y_i$ is a neighbor of $x_i$ in $G_i$, chosen adversarially} \\

 (y_1, \ldots, y_D)   \\

\end{array}
\]

 We now describe the decoder of ${\cal E}$ corresponding to $G$. It first calls $\dec_{G_i, \rho_i}$ on $y_i$ (the decoder corresponding to the Hamming channel with graph $G_i$, and using randomess $\rho_i$), for all $i \in \{1, \ldots,  D\}$. With probability $(1-\epsilon)$, $\dec_{G_i, \rho_i}(y_i)$ correctly returns $x_i$. Thus the expected number of errors (from all inner decoders) is bounded by $\epsilon \cdot D$. It follows that the probability that the number of errors is larger than $2 \epsilon \cdot D$ is at most $e^{-2 \epsilon^2 D}$ (by Chernoff bounds). We choose $\epsilon$ so that $2 \epsilon \leq \beta_{\text{Spielman}}$.  It follows that with probability $1- e^{-2 \epsilon^2 D}$, the fraction of errors made by the inner decoders is smaller than the fraction of errors that can be corrected by Spielman's code. Therefore the decoding procedure of Spielman's code returns  the correct message with high probability.

We next evaluate the runtime of the encoder and the decoder of ${\cal E}$.  The encoder calls the encoder of the outer code, which runs for $O(nD)$ steps. Next, it calls $D$-times the inner encoder, and each one runs in $O(n)$ sterps. Thus the total time is $O(nD)$, that is the encoder runs in linear time in the bit-length of the codeword. The decoder calls the $D$ inner decoders, and each one runs in $2^n \cdot n$ steps. Then it calls the decoder of the outer code, which runs in $O(nD)$. Therefore, taking into account that $D = 2^n$, the decoder runs in time $O((nD)^2)$, that is in quadratic time in the bit-length of the codeword.

\subsection{The condenser from Theorem~\ref{t:guv} is linear}
\label{s:linearguv}

We observe that a very minor modification of the condenser constructed by Guruswami, Umans, and Vadhan in~\cite[Theorem 4.3, also Theorem~1.7]{guv:j:extractor} converts it into a linear condenser. A similar, but more general  modification  (because it works for finite fields of arbitrary characteristic, while our version is for characteristic $2$) has been made by Cheraghchi and Indyk~\cite{che-ind:j:linearcondenser}.

The condenser $C(f,y)$ from~\cite{guv:j:extractor} is viewing the first argument as a polynomial $f \in \F_q$ of degree at most $n-1$, where $q = 2^t$, so $\F_q$ is a field of characteristic $2$. More precisely if $f(Z) = f_0 + f_1 Z + \ldots + f_{n-1}Z^{n-1}$, then the first argument of the condenser is $ (f_0, \ldots, f_{n-1})$, which is represented as a binary string of length $nt$. The second argument $y$ is an element of $\F_q$, thus a binary string of length $t$. The condenser is also using $E[Z]$,  an irreducible polynomial of degree $n$ over $\F_q$, and  a parameter $h$ which can be taken to be a power of $2$. (Note: Requiring $h$ to be a power of $2$ is the only modification from the version in~\cite{guv:j:extractor}.)

The condenser is defined as
\begin{equation}\label{e:guvcondenser}
C(f, y)  = [ y, f(y), (f^h \bmod E)(y),   (f^{h^{2}} \bmod E)(y),  \ldots,  (f^{h^{m-1}} \bmod E)(y)].
\end{equation}

We need to show that each $(f^{h^i}(Z) \bmod E(Z))(y)$ is linear in $f$.

Let us fix $y \in \F_q$ and  consider (for some aribitrary $i$) $A_y : (\F_2)^{nt} \mapping  (\F_2)^t$, defined by
\[
A_y(f)  = [f^{h^i}(Z) \bmod E(Z)](y) \\
= [(f_0 + f_1 Z + \ldots + f_{n-1} Z^{n-1}) ^{h^i} \bmod E(Z)](y) 
\]
It is known that if $a$ and $b$ are elements of a field of characteristic $p$ and $h$ is a power of $p$,
\[
(a+b)^{h^i} = a^{h^i} + b^{h^i}
\]
In our case, the coefficients of the polynomials belong to $\F_q$, which has characteristic $2$, and $h^i$ is a power of $2$. Therefore,
\[
\begin{array}{ll}
A_y(f+g) &=  \big[(f+g)^{h^i}(Z) \bmod E(Z)\big](y)  \\ \\
               & = \big[\big((f_0+g_0) + (f_1+g_1)Z + \ldots + (f_{n-1}+g_{n-1})Z^{n-1}\big)^{h^i} \bmod E(Z)\big](y) \\ \\
& = \big[\big( (f_0+g_0)^{h^i}+  (f_1+g_1)^{h^i}Z^{h^i} + \ldots + (f_{n-1}+g_{n-1}\big)^{h^i}Z^{(n-1)h^i}) \bmod E(Z)\big](y) \\ \\
& = \big[\big(f_0^{h^i}+  f_1^{h^i}Z^{h^i} + \ldots + f_{n-1}^{h^i}Z^{(n-1)h^i}\big) \bmod E(Z)\big](y)  \\ \\
& \quad\quad  +\big[\big(g_0^{h^i}+  g_1^{h^i}Z^{h^i} + \ldots + g_{n-1}^{h^i}Z^{(n-1)h^i}\big) \bmod E(Z)\big](y) \\ \\
& = \big[\big(f_0 + f_1 Z + \ldots + f_{n-1} Z^{n-1}\big) ^{h^i}\ \bmod E(Z)\big](y) \\  \\
& \quad \quad  +  \big[\big(g_0 + g_1 Z + \ldots + g_{n-1} Z^{n-1}\big) ^{h^i}\bmod E(Z)\big](y) \\ \\
& = A_y(f) + A_y(g).
\end{array}
\]
Thus, each component of $C(f,y)$ from the equation~\eqref{e:guvcondenser} is linear and therefore for each  $y$  there exists a $mt$-by-$nt$ matrix $H_y$ with entries in 
$\F_2$ such that $C(f,y) = H_y f$.


\subsection{The condenser from Theorem~\ref{t:tu} is linear.}
 \label{s:lineartu}

The condenser constructed by Ta-Shma and Umans has a structure similar to the condenser from~\cite{guv:j:extractor}. The first argument is denoted $f$ and the second argument (the ``seed") is a pair $(x,y)$.  They use a parameter $h$ which is a power of a prime number $p$. 

The first argument is a polynomial of two variables of the form
\begin{equation}\label{e:polytu}
f(X,Y)  = \sum_{i=0, \ldots, n-1, j=0,1} \alpha_{i,j} X^i Y^j, 
\end{equation}
where the coefficients are in $\F_h$, the fields with $h$ elements, multiplications in $Y$ are modulo an irreducible polynomial $P(Y)$ of degree $2$ 
(formally, $Y$ is an element of $\F_h[X]/P(X)$)  with coefficients in $\F_h$, and the multiplications in $X$ are done modulo an irreducible polynomial $E(X)$ with coefficients in $\F_q = \F_h[Y]/p(Y)$ (formally, $X$ is an element of $\F_q[X]/E(X)$). 

The second argument consists of $x \in \F_q$ and $y \in\F_h$. The condenser is defined by
\[
C(f, (x,y)) = (C_0(f)(x,y), \ldots, C_{m-1}(f)(x,y)),
\]
where each component consists of a polynomial $C_i(f)$ with variables $X$ and $Y$ evaluated at $X=x$ and $Y=y$. The polynomial $C_i(f)$ has the form
\[
\alpha_0 f + \alpha_1 f^h + \ldots + \alpha_{m-1} f^{h^{m-1}}.
\]
For any two polynomials $f_1, f_2$ of the form~\eqref{e:polytu} and for all $\ell$, we have $(f_1+f_2)^{h^\ell} = f_1^{h^\ell} +  f_2^{h^\ell}$, because the polynomials have coefficients in a field with characteristic $p$ and $h^\ell$ is a power of the same prime number $p$, and therefore we can use the same argument as in Section~\ref{s:linearguv}. This implies that $C(f_1+f_2, (x,y)) = C(f_1,(x,y)) + C(f_2, (x,y))$.

\subsection{The condenser  from Theorem~\ref{t:rrv} is linear.}
 \label{s:linearrrv}

The condenser $C_{\rrv} :\zo^n \times \zo^{d}  \rightarrow \zo^{m}$ from~\cite{rareva:j:extractor} (which is actually an extractor) views the first argument $x$ as the specification of a function $u_x(\cdot,\cdot)$ of two variables (in a way that we present below), defines some functions $g_1(y), h_1(y), \ldots, g_m(y), h_m(y)$ (where $y$ is the second argument), each one computable in time polynomial in $n$, and then sets 
\begin{equation}
\label{e:linext}
C_{\rrv}(x,y) = u_x(g_1(y), h_1(y)), \ldots, u_x(g_m(y),h_m(y)),
\end{equation} i.e., the $i$-th bit is
$u_x(g_i(y),h_i(y))$. Thus, it is enough to check that $f_{v,w}(x) = u_x(v,w)$ is linear in $x$. Let us now describe $u_x$. The characteristic sequence of $u_x$ is the Reed-Solomon code of $x$. More precisely,  for some $s$, $x$ is viewed as a polynomial $p_x$ over the field $\F_{2^s}$. The elements of $\F_{2^s}$ are viewed as $s$-dimensional vectors over $\F_2$ in the natural way. Note that in this view the evaluation of $p_x$ at point $v$ is a linear transformation of $x$, i.e.,  $p_x(v) = B_v x$ for some $s$-by-$n$ matrix $B_v$  with entries from $\F_2$. Finally, $u_x(v,w)$ is defined as the inner product $w \cdot p_x(v)$ and therefore $u_x(v,w) = (w B_v)x$, and thus it is a linear function in $x$. Now we plug $h_i(y)$ as $w$ and $g_i(y)$ as $v$, and we build the matrix $A_y$, by taking its $i$-th row to be $h_i(y) B_{g_i(y)}$. Using the Equation~(\ref{e:linext}), we obtain that $C_{\rrv}(x,y) = A_y \cdot x$.

\subsection{On universal codes without shared random bits}
\label{s:nonshared}

  There exists a code that is $(t,\epsilon)$-resilient in the oblivious scenario with rate $1-2t/n-o(1)$ and does not use shared randomness.
  The encoding function uses a large amount of randomness, but it is not shared with the decoder. Unfortunately, the code is not explicit.
\begin{theorem}\label{th:nonexplicitCodes}
  There exists an encoding function $\enc \colon \mcM \times \mcR \rightarrow \mcX$ with $n \le 2t + k + O(\log (tk/\epsilon))$ 
  that is $(t,\eps)$-resilient in the oblivious scenario without shared randomness, where $n = \log \# \mcX$ and $k = \log \# \mcM$.
\end{theorem}

\noindent
A similar result can be shown for the Hamming scenario, but in this case the rate is: $\tfrac{1}{2}-\tfrac{t}{n} - o(1)$.

\begin{proof}
  We show that a random code satisfies the properties. First we define some parameters for later reference.
  Let $T = 2^t$ and $M = \# \mcM$ (which equals $2^k$). Let $N$ be such that 
  \[
    N = 2RMT/\epsilon, \quad \text{with} \quad R = \tfrac 6 \eps (\log (MT) + T \log N).
  \]
  
  For each $\rho \in \mcR$ and $m \in \mcM$, select $\enc_\rho(m)$ randomly in $\mcX$. 
  We show that for all sets $E$ of size $T$, all $m \in \mcM$ and all $e \in E$, we have
  \begin{equation}\label{e:goal}
    \# \left(\enc_\mcR(m) + e \right) \cap \left(\enc_\mcR (\mcM \setminus\{m\}) + E \right) \;\,\le\;\, \epsilon R,
  \end{equation}
  where for $\mcM' \subseteq \mcM$,
  \[
    \enc_\mcR(\mcM') + E \;=\; \left\{\enc_{\rho'}(m') + e' : \rho' \in \mcR, m' \in \mcM', e' \in E \right\}
  \]
  and similar for $\enc_\mcR(m)+e$.
  This implies that given a set $E$, we obtain a decoding function $\dec$ that satisfies~\eqref{eq:oblivious} using a greedy search, 
  i.e., $\dec(\tx)$ is the first message $m$ that appears in a search for triples $(\rho,m,e)$ that satisfy $\enc_\rho(m) + e = x$.
  It remains to prove that a random function $\enc$ satisfies this property.
  
  Fix some $E,m,e$ in the condition~\eqref{eq:oblivious}.
  Assume we have already randomly selected $\enc_\mcR(m')$ for all $m' \not= m$, and we will now select $\enc_\rho(m)$ for all~$\rho \in \mcR$.
  The probability that a random $x \in \mcX$ belongs to $B := \enc_\mcR (\mcM \setminus\{m\}) + E$
  is at most~$MT/N$. This is bounded by $\epsilon/(2R)$ by choice of~$N$.  Therefore, if we choose $R$ many such elements $x$ at random, the expected number of them that fall in $B$ is at most $\epsilon/2$. 
  The probability that more  than $\epsilon R$ elements are in $B$  is at most
  \[
    \exp(-\epsilon R/6)
  \]
  by the Chernoff bound in multiplicative form. 
  The probability that this happens for some set $E$ of size $T$, some element $m \in M$ and some $e \in E$, is at most
  \[
    N^T \cdot M \cdot T \cdot \exp(-\epsilon R/6)
  \]
  by the union bound.
  By the choice of $R$,  this is less than~$1$.
  Hence, with positive probability the conditions are satisfied, and in particular the encoding function exists.
\end{proof}

\subsection{The additive Hamming scenario:  intermediate resilience between the oblivious and the Hamming scenarios}\label{s:weakhamming}

In the oblivious scenario, a universal code is resilient to channels that do not have access to the transmitted codeword $x$ and that corrupt it by  adding a noise vector $e$ from a fixed set $E$.
In the Hamming scenario, a universal code is resilient to channels that have access to the transmitted codeword $x$ and that  corrupt it by adding   a  noise vector $e$ which depends on $x$ (because $e=x + \tx$, where $\tx$ is a neighbor of $x$ in the bipartite graph). 
In this section we consider universal codes that are resilient to an intermediate type of channels: they have access to $x$ (like in the Hamming scenario), but can only add $e$ from a fixed set $E$ (like in the oblivious scenario).
In other words, in the additive Hamming scenario we use the same definition of resilience applied to Hamming channels of the following special form. 
The bipartite graphs have left set $\mcX$ and right set $\mctX = \mcX$, and the edges are given by
\[
  \bigcup_{x \in \mcX} \{x\} \times \left(\{x\} + E\right),
\]
for some set~$E$.

\begin{figure}[ht]
\centering
\begin{tikzpicture}[->,>=stealth',shorten >=1pt,scale=0.8,auto,node distance=3cm, transform shape]
 \node[] (1) at (1,4) {$\rho$};
  \node[] (2) at (1,2) {$m$};
  \node[fNode](4) at (3,2) {$\enc$};
  \node[rNode](5) at (5,2) {$\cha$};
  \node[fNode](6) at (8.5,2){$\dec$};
  \node(7)[] at (10.5,2) {$m$};
 \path[every node/.style={font=\sffamily\small}]
    (1) edge node  [left]{}(4)
    (2)    edge node [left]{}(4)
(4) edge [] node [above]{$x$} (5)
      (5) edge   node [above]{$\tx = x + \cha(x)$} (6)
        (6) edge (7);
  \draw (1) -| (6);
\end{tikzpicture}
  \caption{The additive Hamming scenario: $\cha(x)$ is chosen from a fixed set $E$ of size $2^t$.}
\label{f:weakhamming}
\end{figure}
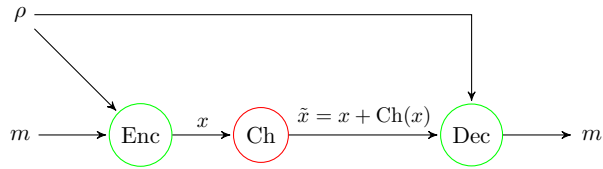

\begin{definition}\label{def:weakHamming} 
  Let $\mcX$ be an additive group.  A private code $\enc \colon \mcM \times \mcR \rightarrow \mcX$ is {\em $(t,\epsilon)$-resilient in the additive Hamming scenario} 
 if for every set  $E \subseteq \mcX$ of size at most $2^t$, 
  there exists a decoding function $\dec \colon \mctX \times \mcR \rightarrow \mcM$ such that
  for all channel functions $\Ch \colon \mcX \mapping E$ and all $m \in \mcM$
  \[
    \Pr_{\rho \in \mcR}  \left[ \dec_\rho( \enc_\rho(m) + \Ch(\enc_\rho(m))) = m \right] \;\, \ge\;\, 1-\epsilon.
    \]
\end{definition}

Recall that there exists a $(t, \epsilon)$ universal code in the oblivious scenario with rate $1-t/n - o(1)$ that uses $O(\log n)$ random bits, and a $(t, \epsilon)$ universal code in the Hamming  scenario with rate $1-t/n - o(1)$ that uses $2n$ random bits.  The question is how many random bits are needed for  a universal code in the additive Hamming scenario to achieve the same rate.

We show that there is a universal code in the additive Hamming scenario that uses $n + O(\log n)$ random bits. This universal code is obtained from the following general result that shows that a code for the oblivious scenario can be transformed into a code for the additive Hamming scenario at the cost of using more shared randomness.


\begin{proposition}\label{prop:}
 If there exists a $(t,\epsilon)$-resilient code $\enc \colon \mcM \times \mcR \rightarrow \mcX$ in the oblivious scenario, then 
  there exists a $(t,\epsilon)$-resilient code $\enc' \colon \mcM \times \mcR' \rightarrow \mcX$ in the additive Hamming scenario with $\mcR' = \mcR \times \mcX$.
\end{proposition}


The encoding function $\enc'$ is obtained by adding a random element from $\mcX$ to $\enc$. 
In this way, for each fixed message $m \in \mcM$, the distribution of added elements by the channel function does not depend on the encoding function, 
and this allows us to apply the oblivious scenario. We present the details.

\begin{proof}
  By definition of $(t,\epsilon)$-resilience in the oblivious scenario, we have that for all $m$ and $e \in E$:
  \begin{equation}\tag{\ensuremath{*}}\label{e:assumption}
    \Pr_\rho \left[\dec_\rho(\enc_\rho(m) + e) = m \right] \;\ge\; 1-\epsilon
  \end{equation}
  The new encoding $\enc'$ evaluates $\enc$ and adds a random $z \in \mcX$, where $\rho' = (\rho,z)$.  Thus $\enc'_{\rho'}(m) = \enc_\rho(m) + z$.
  The new decoding $\dec'$ function subtracts $z$ and evaluates $\dec$, thus $\dec'_{\rho'}(\tilde{x}) = \dec_\rho(\tilde{x}-z)$.

  We now verify the requirement in Definition~\ref{def:weakHamming} with $\rho' = (\rho,z)$. 
  Let $\cha$ be a channel function and let $m \in \mcM$.
  Note that since $\enc'$ adds $z$ at the end of its evaluation, and $\dec'$ starts by subtracting $z$, the 
  variable $z$ only appears in this condition through the channel function in the term
  \[
    \cha(\enc_\rho(m) + z).
  \]
  Since $\rho$ and $z$ are independent, the value of $\enc_\rho(m) + z$ is uniform in~$\mcX$ and independent of $\rho$.
  Hence, we could replace the above quantity by $\cha(z)$.
  Let $\xi = \cha(z)$ be this random variable, which has values in~$E$.
  The condition of  definition~\ref{def:weakHamming} can now be written as
  \[
    \Pr_{\rho,\xi} \left[\dec_\rho(\enc_\rho(m) + \xi) = m \right] \;\ge\; 1-\epsilon
  \]
  This is a convex combination of \eqref{e:assumption} for various $e \in E$, and 
  hence the inequality holds.
\end{proof}

\if01
\begin{proposition}[Invertible function $\rightarrow$ code in the additive Hamming scenario]
\label{l:sh3}
Let $(F_n)$ be a family of functions such that for every $n$, $F_n :\zo^n \times \zo^{d_n} \mapping \zo^{t_n+\Delta_n}$ is linear,  $(t_n, \epsilon_n)$-invertible and $t_n + \Delta_n \leq n$.
 \smallskip

Then there exists a private code $\enc$  that is  $(\tk_n, \epsilon_n)$-resilient in the additive  Hamming scenario, with rate $1- (\tk_n+\Delta_n)/n$, and such that the encoder and the decoder share  $n+d_n$ random bits.
\end{proposition}

\begin{proof}
The proof in Proposition~\ref{p:sh2} does not work for the additive Hamming scenario, because now the noise $e$ depends on $\rho$ and consequently we cannot infer that the inverter $g$ finds $e$ with high probability.  We make a modification:  when  $\enc$ builds the codeword, he  adds a random mask $a$,  so that the codeword and, consequently also the noise $e$,  are independent of $\rho$. 

We next present the details. We use the same notation as in Proposition~\ref{p:sh2} (thus we drop the subscript $n$, $H_\rho$ is the matrix defining the linear function $F(\cdot, \rho)$, etc.).
The encoding and decoding procedures are as follows:
\begin{enumerate}
\item The encoder $\enc$  and the decoder $\dec$  share a random string $\rho \in \zo^d$ and a random string $a \in \zo^n$.
\item Consider a fixed message $m$  of length  $n - (\tk + \Delta)$. $\enc$, using $\rho$ and $a$,  computes the codeword $x$ of length $n$ as follows:
\[ x = a + (m \mbox{-th  element in the null space of } H_\rho).
\]
 Note that the dimension of the null space of $H_\rho$ is at least  $n - (\tk + \Delta)$, because the rank of $H_\rho$ is at most $\tk + \Delta$.  Thus the encoder is well defined.

Also, note that for every $\rho_0 \in \zo^d$, the distribution of $x$, conditioned on $\rho=\rho_0$, is the uniform distribution over $(\F_2)^n$. Thus, $x$ and $\rho$ are independent. 

\item The decoder $\dec$ has $\rho$ and $a$ (shared with the encoder) and he receives  $\tx = x+e(x) $, where  $e(x)  = \cha(x) $ is the noise added by the channel, and which belongs to the fixed set $E$ of size at most $2^t$ that defines the channel. We have  $x-a \in {\rm NULL}(H_\rho)$ and thus $H_\rho \tx = H_\rho(x - a + a +e(x)) =H_\rho a + H_\rho e(x)$.

$\dec$ computes  $p = H_\rho \tx - H_\rho a$.  By the above observation, $F(e(x), \rho) = H_\rho e(x)  = p$.  Now, we use the inverter $g$ corresponding to $F$ for the list of suspects $E$. More precisely, $\dec$ runs the inverter $g$ on input $F(e(x), \rho)$  \marius{and list $E$,}  and prints out the output of $g$.
Thus, the probability that $\dec$ outputs $e(x)$ is equal to \marius{[modifications below,  because the syntax of g has changed.]}
\begin{equation}\label{e:prob1}
\prob_{\rho,a} [g_E (F(e(x), \rho)) = e(x)].
\end{equation}
This probability is equal to 
\[
\sum_{e \in E} \prob [g_E (F(e, \rho))= e \wedge e=e(x)] = \sum_{e \in E} \prob  [g_E(F(e, \rho))= e] \cdot \prob[e=e(x)].
\]
The above equality holds because $x$ and $\rho$ are independent and the event ``$g_E(F(e, \rho)) = e$" does not depend on $a$.
The inverter $g$ guarantees that for every $e \in E$,  $\prob [g_E(F(e, \rho))= e]  \geq 1-\epsilon$. It follows that the probability in~\eqref{e:prob1}  is also at least $1-\epsilon$.

\end{enumerate}
\end{proof}
Now, if we plug the invertible function from Theorem~\ref{t:inv} in the above lemma, we obtain the following result.
\begin{theorem}\label{t:weakhcode}
There exists a universal code in the additive Hamming scenario with the same parameters as the universal code for the oblivious scenario in Theorem~\ref{t:main}, except that the number of shared random bits is $n + O(\log n + \log(1/\epsilon))$.
\end{theorem}  
\fi
Thus, a resilient code in the oblivious scenario for $\mcX = \{0,1\}^n$ that uses $r$ bits of shared randomness, can be transformed into a code for the additive Hamming scenario that has the same rate and uses $n+r$ bits of randomness.
In particular, there exist $(t,\epsilon)$-resilient codes that use $d = n+O(\log n)$ randomness by Theorem~\ref{t:main}. 
This is better than the value $2n$ in Theorem~\ref{t:bruno} for general Hamming channel,  but is far from the $O(\log n)$ value achieved by the universal code for the oblivious scenario from Theorem~\ref{t:main}. In the next theorem, we provide a lower bound, which shows that the number of random bits has to be $\Omega(n)$ for values $k, t = \Omega(n)$ (which  are typical in most applications).

\begin{theorem}\label{p:hamlb}
  Let $\epsilon<1$. If a code $\enc\colon \zo^k \times \zo^d \rightarrow \zo^n$ is $(t, \epsilon)$-resilient in the additive Hamming scenario, then $d \ge \min\{k,t/2\}-O(1)$. 
\end{theorem}

\begin{proof}
  Fix some encoding function $\enc$ that uses $d$ random bits, and assume
  \[
    d \le \min\{k,t/2\}-c 
  \] 
  for some constant $c$. We show that for large $c$, the code $\enc$ can not be $(t,\epsilon)$-resilient.
  We use the incompressibility method. Let $C(x)$ denote the Kolmogorov complexity of a string~$x$, which is
  given by the minimal length of a program that outputs~$x$. For this we need to fix a programming language, and we 
  choose a language that makes the complexity function minimal up to an $O(1)$ constant, see~\cite{suv:b:kolmenglish,liVitanyi2019} for more background.
  In our proof, we assume that the programs have an extra input which is a description of the encoding function~$\enc$, (but we omit this in the notation). We also assume that $t/2 \le k$ (otherwise, we just take in the arguments below $t = 2k$).
  
  Let $0^n$ be the string containing $n$ zeros. Note that $C(0^n) \le O(1)$, because $n$ is a parameter of $\enc$, to which our programs have access.
  We consider a channel  in the additive Hamming scenario that for codewords of length $n$ adds a noise vector $e \in E = \{u \in \zo^n \mid C(u) < t\}$. 
  Note that $E$ has size smaller than $2^t$ and, thus,  this channel has distortion at most~$2^t$.

  Consider a message $m$ with 
  \[
    t/2 \le C(m) \le  t/2 + O(1).
  \]

  Such a message exists, since the total number of messages is~$2^k$ and~$t/2 \le k$.
  Note that for all choices of randomness $\rho \in \zo^d$ we have\footnote{
    We prove this by concatenating $\rho$ to a program for $m$. 
    From this concatenation we can always retrieve back the splitting point, since $\rho$ has length $d$, which is a parameter of~$\enc$.
    }
  \[
    C(\enc_\rho(m)) \;\le\; C(m) + \text{length}(\rho) + O(1) \;\le\; t-c + O(1)  \; < \;  t,
  \]
  where the last inequality holds for a sufficiently large large $c$. Thus if we denote $x = \enc_\rho(m)$, then $x \in E$, and the channel can add the noise vector $x$ to $x$ obtaining $\tx=x+x = 0^n$.
  On the other hand, 
  \[
    C(\dec_\rho(0^n)) \;\le\; \text{length}(\rho) + O(1) \;\le\; t/2 - c + O(1)   \; < \;  C(m)
  \]
  Thus, for large $c$ and for all $\rho$, we have $\dec_\rho(0^n) \not= m$. 


  In other words, $\dec_\rho(\enc_\rho(m) + \cha(\enc_\rho(m))) \not= m$ for all~$\rho$, thus $\enc$ is not $(t,\epsilon)$-resilient for $\eps<1$.
\end{proof}

\if01
\subsection{Proofs of Theorem~\ref{t:randomdistort} and Theorem~\ref{t:piecewiseH}}\label{s:proofranddistort}

\textbf{Proof of Theorem~\ref{t:randomdistort}.}  The construction uses the Justesen concatenation scheme, which combines an \emph{outer code}, with several \emph{inner codes.} First the message written with symbols from  the alphabet $\F_2^k$ (for some integer $k$)  is mapped by the encoder of the outer code  into a codeword over the same alphabet, and next each symbol of the codeword is mapped using the encoder of one the inner codes.

\emph{The outer code:}  Following Cheraghchi~\cite{che:c:codescondenser}, we use the linear time encodable/decodable  code constructed by Spielman~\cite{spi:j:code} as the outer code.
\begin{theorem}[~\cite{spi:j:code}]
\label{t:spielman}
For every constant $\alpha < 1$ and every positive integer $k$, there exist a constant $\beta_{\text{Spielman}} > 0$ and an explicit family of codes $(C_S)_{S \in \nat}$ over the alphabet $\F_2^k$, such that $C_S$ encodes messages of length $S$,   has rate $1-\alpha$,  and is error correcting for a fraction $\beta_{\text{Spielman}}$  of errors.The encoder and the decoder run in time that is linear in the bit-length of the codewords.  

More explicly,  for every $S \in \nat$, the encoder of $C_S$ maps $(\F_2^k)^S$ into $(\F_2^k)^D$, with $S/D \geq 1 - \alpha$,   and for every codeword $x \in (\F_2^k)^D$, and every $\tx$ such that that the relative Hamming distance between $x$ and $\tx$ is $\beta_{\text{Spielman}}$, the decoder on input $\tx$ returns the message encoded as $x$. The encoder and the decoder run in time $O(kD)$.
\end{theorem}

\emph{The inner codes:} The inner codes are obtained from the universal code from Proposition~\ref{p:sh1}. Recall that the encoder of this code is a function $\enc (m, H)$, where $H$ is a random 
$(t + \log(1/\epsilon)) \times n$ matrix $H$ with elements in $\F_2$. For $k=n - (t+ \log(1/\epsilon)$,  the encoder maps a $k$-bit message $m$ into an $n$-bit codeword,  which is the $m$-th element of the null space of $H$ (using some canonical ordering of the elements in the null space).  The value of $\epsilon$ will be picked later.  Let $d = n (t+\log (1/\epsilon))$. We use, as inner codes,  $D = 2^d$ codes $\enc_{H_1}, \enc_{H_2}, \ldots, \enc_{H_D}$, which are obtained from $\enc$ in Proposition~\ref{p:sh1}, by fixing the randomness to every $d$-bits string, i.e., for each $H \in \zo^d$,  $\enc_{H}(\cdot) \stackrel{def.}{=} \enc(\cdot, H)$. Each inner code maps a $k$-bit string, viewed in the natural way as an element of $\F_2^k$  into an $n$-bit string, which similarly is viewed in the natural way as an element of $\F_2^n$.

Equipped with  the outer code and the inner codes, we construct a new encoder ${\cal E}$  as  follows.

\emph{Concatenation scheme:} First, the $S$-symbol input message  $(m_1, \ldots, m_S)$ is encoded with the outer code into a $D$-symbol word $(c_1, \ldots, c_D)$.
Next, each $c_i$ is encoded with the inner code $\enc_{\rho_i}$.
\[
\begin{array}{l}
 (m_1, \ldots, m_S)  \\

\quad \quad \bigg \downarrow   \text{\quad \quad\quad outer code encoding;  each $m_i$ and $c_i$ is $k$-bits long } \\

 (c_1, \ldots, c_D)   \\

\quad \quad \bigg \downarrow \text{\quad \quad  \quad inner codes encoding } \\

 (\enc_{\rho_1}(c_1), \ldots,  \enc_{\rho_D}(c_D))
\end{array}
\]

Thus ${\cal E}$ maps binary strings of length $k\cdot  S$ into codewords that are binary strings of length  $n \cdot D$.

Consider now a $D$-memoryless channel $\cha$ in the oblivious scenario adjusted def for the oblivious scenario  that has distortion bounded by $t$. 
Recall that this means that: 
\begin{enumerate}
\item  there is  a set $E \subseteq \F_2^n$ of size $T = 2^t$, and 
\item  $\cha$ takes as  input a $D$-tuple $(x_1, \ldots, x_D) \in (\F_2^n)^D$ , and outputs $(y_1, \ldots, y_D)$, where each $y_i = x_i+e_i$,  for $e_i$  chosen uniformly at random in $E$,  independently of the other choices.  
\end{enumerate}

Suppose the sender encodes the message $(m_1, \ldots, m_S)$ into  the codeword $(x_1, \ldots, x_D)$ (where each $x_i \in \F_2^n$) and the channel distorts it into $(y_1, \ldots, y_D)$. 
\[
\begin{array}{l}
(m_1, \ldots, m_S) \\

\quad \quad \bigg \downarrow   \text{\quad \quad\quad  encoding with the concatenated  outer/inner code } \\

 (x_1, \ldots, x_D)  \\

\quad \quad \bigg \downarrow   \text{\quad \quad\quad channel distortion; $y_i=x_i + e_i$, for $e_i$ randomly chosen in $E$ } \\

 (y_1, \ldots, y_D)   \\

\end{array}
\]

Henceforth, we consider that $(x_1, \ldots, x_D)$ is fixed, and $(y_1, \ldots, y_D)$ is a random variable depending on the randomness of the channel. The decoder of the concatenated code, first calls the decoders of the $D$ inner codes respectively on each component of $(y_1, \ldots, y_D)$, which return the word $(z_1, \ldots, z_D)$, and next calls the decoder of the outer code on this latter word. We show that with high probability, this procedure reconstructs $(m_1, \ldots, m_S)$.

For each $e \in E$, we say that the matrix $H$ is \emph{good for $e$}  if $H e_1 \not= He$ for all $e_1 \in E - \{e\}$. In the proof of Proposition~\ref{p:sh1},   it is shown that if $H$ is good for $e$, then the decoder on input $x+e$ reconstructs $x$ and that for each $e \in E$, at least a fraction of $(1-\epsilon)$ of $H$'s are good for $e$.  By a standard averaging argument, it follows that  there is a set (of ``good" matrices) GOOD$_{\text{rand}}$ containing $(1-\sqrt{\epsilon})$ fraction of $H$'s  in $\zo^d$ and  a set (of ``good" noise vectors)  GOOD$_{\text{noise}}$  containing $(1-\sqrt{\epsilon})$   fraction of $e$'s  in $E$, such that every $H$ in GOOD$_{\text{rand}}$ is good for every $e$ in  GOOD$_{\text{noise}}$. By rearranging the tuple $(y_1, \ldots, y_D)$, we can assume that the first $(1-\sqrt{\epsilon})D$ components correspond to the ``good" matrices.  On every  $i$ in this segment of ``good" components,  if  $e_i =y_i - x_i$  is a ``good" noise vector,  the inner decoder on input $y_i$ correctly reconstructs $x_i$. Note that the probability (over the randomness of the channel)  that $e_i$ is a ``good" noise vector  is at least $1 - \sqrt{\epsilon}$.

Let $\mu$ be  the expected number of components of $(y_1, \ldots, y_D)$ on which the inner decoders are incorrect. Note that 
\[
\mu \le \sqrt{\epsilon} \cdot D + (1-\sqrt{\epsilon}) \cdot D \cdot \sqrt{\epsilon}.  
\]
In the sum above, the first term  corresponds to the errors made by the inner decoders in the  bad segment and the second term corresponds to the errors made by decoders in the good segment  for which the channel is choosing a bad $e$.  Thus, $\mu < \mu_H$, where $\mu_H \stackrel{def.}{=} 2 \sqrt{\epsilon}\cdot D$. We take $\gamma = \sqrt{\epsilon} \cdot D$.
\[
\prob [ \# \text{ errors } \ge 3 \sqrt{\epsilon} \cdot D] = \prob [ \# \text{ errors } \ge \mu_H + \gamma] \le e^{-2 \gamma^2/D} = e^{-2\epsilon \cdot D}. 
\]
 We have used a form of the Chernoff bounds~\footnote{Let $X = \sum_{i=1}^D X_i$, where $X_i$, $i \in \{1,D\}$ are independently distributed in $[0,1]$. Let $\mu_H \geq \mu$, where $\mu$ is the expected value of $X$. Then for every $\gamma > 0$, $\prob[X \geq \mu_H + \gamma] < e^{-2 \gamma^2/D}$.} for the case when we know an upperbound of the expected value (see Exercise 1.1 (a) in~\cite{dub-pan:b:concentration})

Take $\epsilon$ such that $3 \sqrt{\epsilon} \le \beta_{\text{Spielman}}$. Then, with probability at least $1 -  e^{-2\epsilon \cdot D}$, the inner decoders are correct on all except at most  
$\beta_{\text{Spielman}}$ fraction of positions. In such a case, the outer decoder is able to reconstruct the codeword $(x_1, \ldots, x_D)$ and then the message that is encoded into this codeword.

We now evaluate the runtime of the encoder and the decoder. The encoder calls first the encoder of the outer code, which takes time $O(nD)$, and then calls the encoders of the $D$ inner codes, each one running $O(n^2)$ steps (by Remark~\ref{r:complexity1}). Thus, the total time for encoding is $O(n^2 D)$, which is quasi-linear in $nD$ (the bit-length of a codeword).  The decoder first calls the $D$ inner deccoders, and each one runs in time $O(T n^2)$, so this step takes $O(DTn^2)$, which is $O((nD)^2)$ because $T = O(D)$. Next, it calls the decoder of the outer code, which runs in time $O(nD)$. So the total time is less than $O((Dn)^2)$, so quadratic in the bit-length of a codeword.

\bigskip

\textbf{Proof   of Theorem~\ref{t:piecewiseH}.}  We use again a concatenation scheme with an outer code and an innner code. The outer code is Spielman's error correcting code~\cite{spi:j:code}, presented in Theorem~\ref{t:spielman}.  The inner code is the code from Theorem~\ref{t:bruno}. Its encoder
$\enc_\rho(m)$ maps a $k$-bit message $m$ into an $n$-bit codeword, and $\rho$ represents the  random string used for both encoding and decoding (we recall that we are in the shared randomness setting). The outer code works with strings having symbols from the alphabet $\Sigma^k$.  The encoder of ${\cal E}$ first calls the encoder of the outer code, which  maps an $S$-symbol message $(m_1, \ldots, m_S)$ into a $D$-symbol codeword $(c_1, \ldots, c_D)$.  Next, each $c_i$ is encoded with the inner code $\enc_{\rho_i}$. The encoder and the decoder of ${\cal E}$ share randomness $\rho=(\rho_1, \ldots, \rho_D)$. 

\[
\begin{array}{l}
 (m_1, \ldots, m_S)  \\

\quad \quad \bigg \downarrow   \text{\quad \quad\quad outer code encoding;  each $m_i$ and $c_i$ is $k$-bits long } \\

 (c_1, \ldots, c_D)   \\

\quad \quad \bigg \downarrow \text{\quad \quad  \quad inner code encoding using randomness $\rho = (\rho_1, \ldots, \rho_D)$ } \\

 (\enc_{\rho_1}(c_1), \ldots,  \enc_{\rho_D}(c_D))
\end{array}
\]


Thus, the encoder of ${\cal E}$ maps a $kS$-bit string into an $nD$-bit string, and consequently has rate $(kS) /(nD) = (S/D) \cdot (k/n) \ge (1-\alpha) \cdot (1-t/n - o(1))$.

Suppose a piecewise Hamming channel  with graph $G = G_1 \times \ldots \times G_D$ distorts $(x_1, \ldots, x_D)$ into $(y_1, \ldots, y_D)$.
\[
\begin{array}{l} 
 (x_1, \ldots, x_D)  \\

\quad \quad \bigg \downarrow   \text{\quad \quad\quad channel distortion; $y_i$ is a neighbor of $x_i$ in $G_i$, chosen adversarially} \\

 (y_1, \ldots, y_D)   \\

\end{array}
\]

 We now describe the decoder of ${\cal E}$ corresponding to $G$. It first calls $\dec_{G_i, \rho_i}$ on $y_i$ (the decoder corresponding to the Hamming channel with graph $G_i$, and using randomess $\rho_i$), for all $i \in \{1, \ldots,  D\}$. With probability $(1-\epsilon)$, $\dec_{G_i, \rho_i}(y_i)$ correctly returns $x_i$. Thus the expected number of errors (from all inner decoders) is bounded by $\epsilon \cdot D$. It follows that the probability that the number of errors is larger than $2 \epsilon \cdot D$ is at most $e^{-2 \epsilon^2 D}$ (by Chernoff bounds). We choose $\epsilon$ so that $2 \epsilon \leq \beta_{\text{Spielman}}$.  It follows that with probability $1- e^{-2 \epsilon^2 D}$, the fraction of errors made by the inner decoders is smaller than the fraction of errors that can be corrected by Spielman's code. Therefore the decoding procedure of Spielman's code returns  the correct message with high probability.

We next evaluate the runtime of the encoder and the decoder of ${\cal E}$.  The encoder calls the encoder of the outer code, which runs for $O(nD)$ steps. Next, it calls $D$-times the inner encoder, and each one runs in $O(n)$ sterps. Thus the total time is $O(nD)$, that is the encoder runs in linear time in the bit-length of the codeword. The decoder calls the $D$ inner decoders, and each one runs in $2^n \cdot n$ steps. Then it calls the decoder of the outer code, which runs in $O(nD)$. Therefore, taking into account that $D = 2^n$, the decoder runs in time $O((nD)^2)$, that is in quadratic time in the bit-length of the codeword.
\fi

\if01
\newpage
 \marius{[proof for the hamming scenario, shared randomness, and random distortion]}

The construction uses the Justessen concatenation scheme, which combines an \emph{outer code}, with several \emph{inner codes.} First the message written over the alphabet $\F_2^k$ (for some integer $k$)  is mapped by the encoder of the outer code  into a codeword over the same alphabet, and next each symbol of the codeword is mapped using the encoder of one the inner codes.

\emph{The outer code:}  Following Cheraghchi~\cite{che:c:codescondenser}, we use the linear time encodable/decodable  code constructed by Spielman~\cite{spi:j:code} as the outer code.
\begin{theorem}[~\cite{spi:j:code}]
For every constant $\alpha < 1$ and every positive integer $k$, there exist a constant $\beta_{\text{Spielman}} > 0$ and an explicit family of codes $(C_S)_{S \in \nat}$ over the alphabet $\F_2^k$, such that $C_S$ encodes messages of length $S$,   has rate $1-\alpha$,  and is error correcting for a fraction $\beta_{\text{Spielman}}$  of errors.The encoder and the decoder run in time that is linear in the bit-length of the codewords.  

More explicly,  for every $S \in \nat$, the encoder of $C_S$ maps $(\F_2^k)^S$ into $(\F_2^k)^D$, with $S/D \geq 1 - \alpha$,   and for every codeword $x \in (\F_2^k)^D$, and every $\tx$ such that that the relative Hamming distance between $x$ and $\tx$ is $\beta_{\text{Spielman}}$, the decoder on input $\tx$ returns $x$. The encoder and the decoder run in time $O(kD)$.
\end{theorem}

\emph{The inner codes:} We use $D = 2^d$ codes $\enc_{\rho_1}, \enc_{\rho_2}, \ldots, \enc_{\rho_D}$, which are obtained from the private code $\enc$ in Theorem~\ref{t:bruno}, by fixing the randomness to every $d$-bits string, i.e., for each $\rho \in \zo^d$,  $\enc_{\rho}(\cdot) \stackrel{def.}{=} \enc(\cdot, \rho)$. Each inner code maps a $k$-bit string, viewed in the natural way as an element of $\F_2^k$  into a $n$-bit string, which similarly is viewed in the natural way as an element of $\F_2^n$.

Equipped with  the outer code and the inner codes, we construct a new encoder $E$  as  follows.

\emph{Concatenation scheme:} First, the $S$-symbol input message  $(m_1, \ldots, m_S)$ is encoded with the outer code into a $D$-symbol word $(c_1, \ldots, c_D)$. Then, we mask each $c_i$ with a random vector $a_i$, obtaining $c'_i = c_i + a_i$. 
Next, each $c'_i$ is encoded with the inner code $\enc_{\rho_i}$.
\[
\begin{array}{l}
 (m_1, \ldots, m_S)  \\

\quad\quad \bigg \downarrow   \text{\quad \quad \quad outer code encoding; } \\

 (c_1, \ldots, c_D) \\

\quad \quad \bigg \downarrow \text{\quad\quad\quad masking}\\

(c_1+a_1, \ldots, c_D + a_D) \\

\quad \quad   \bigg \downarrow \text{\quad \quad  \quad inner codes encoding  }  \\

 (\enc_{\rho_1}(c_1+a_1), \ldots,  \enc_{\rho_D}(c_D+a_D))
\end{array}
\]

Thus $E$ maps binary strings of length $k\cdot  S$ into codewords that are binary strings of length  $n \cdot D$.

Consider now a $D$-memoryless channel $\cha$ that has distortion bounded by $t$. 
Recall that this means that: 
\begin{enumerate}
\item  there is  a bipartite graph $G = (L=\zo^n, R, E \subseteq L \times R)$ and all right degrees are bounded by $T = 2^t$, and 
\item  $\cha$ takes as  input a $D$-tuple $(x_1, \ldots, x_D)$ of left nodes, and outputs $(y_1, \ldots, y_D)$, where each $y_i$ is chosen uniformly at random  among the right neighbors of $x_i$, and the $D$ random choices are independent.
\end{enumerate}

Suppose the sender transmits the codeword $(x_1, \ldots, x_D)$ (where each $x_i \in \F_2^n$) and the channel distorts it into $(y_1, \ldots, y_D)$. The decoder of the concatenated code, first calls the decoders of the $D$ inner codes respectively on each component of $(y_1, \ldots, y_D)$, which return the word $(z_1, \ldots, z_D)$, and next calls the decoder of the outer code on this latter word. We show that with high probability (over the randomness used by the channel), this procedure reconstructs $(x_1, \ldots, x_D)$.

For each edge $(x,y)$ of the graph $G$, we say that $(a,\rho)$ is \emph{good}  if  at most one element of  $a+ \{\enc(m, \rho) \mid \text{$m$ message}\}$ is a neighbor of $y$ (i.e.,   at most one codeword of $\enc_\rho$ shifted by $a$ is a neighbor of $y$). The analysis in the proof of Theorem~\ref{t:bruno} shows that   for each edge $(x,y)$, $(1-\epsilon)$ fraction of pairs $(a, \rho)$'s are good. By a standard averaging argument, it follows that  there is a set (of ``good" $(a, \rho)$'s) GOOD$_{\text{rand}}$ containing $(1-\sqrt{\epsilon})$ fraction of $\rho$'s  in $\zo^n \times \zo^d$ and  a set (of ``good" edges)  GOOD$_{\text{edge}}$  containing $(1-\sqrt{\epsilon})$   fraction of edges   of $G$, such that every $(a, \rho)$ in GOOD$_{\text{rand}}$ is good for every edge $(x,y)$ in  GOOD$_{\text{edge}}$. By rearranging the tuple $(y_1, \ldots, y_D)$, we can assume that the first $(1-\sqrt{\epsilon})D$ components correspond to the ``good" $\rho$'s.  On every  $i$ in this segment of ``good" components,  if  $y_i$ is ``good,"  the inner decoder is correct. Note that the probability (over the randomness of the channel)  that $y_i$ is ``good"  is at least $1 - \sqrt{\epsilon}$.  \marius{There's a problem here. It may happen that most (or even all) neighbors of $x_i$ are ``bad,"  and then the last sentence is wrong.  I think the argument can be salvaged for the oblivious scenario: Instead of looking at a right node $y$, we look at a noise vector $e$ in $E$, etc.} 

Let $\mu$ be  the expected number of components on which the inner decoders are incorrect. Note that 
\[
\mu \le \sqrt{\epsilon} \cdot D + (1-\sqrt{\epsilon}) \cdot D \cdot \sqrt{\epsilon}.  
\]
In the sum above, the first term  corresponds to the errors made by the inner decoders in the  bad segment and the second term corresponds to the errors made by decoders in the good segments for which the channel is choosing a bad $y$.  Thus, $\mu < \mu_H$, where $\mu_H \stackrel{def.}{=} 2 \sqrt{\epsilon}\cdot D$. We take $\gamma = \sqrt{\epsilon} \cdot D$.
\[
\prob [ \# \text{ errors } \ge 3 \sqrt{\epsilon} \cdot D] = \prob [ \# \text{ errors } \ge \mu_H + \gamma] \le e^{-2 \gamma^2/D} = e^{-2\epsilon \cdot D}. 
\]
 We have used a form of the Chernoff bounds for the case when we know an upperbound of the expected value (see Exercise 1.1 (a) in~\cite{dub-pan:b:concentration}).

Take $\epsilon$ such that $3 \sqrt{\epsilon} \le \beta_{\text{Spielman}}$. Then, with probability at least $1 -  e^{-2\epsilon \cdot D}$, the inner decoders are correct on all except at most  
$\beta_{\text{Spielman}}$ fraction of positions. In such a case, the outer decoder is able to reconstruct $(x_1, \ldots, x_D)$.

\fi

\if01
\newpage

\marius{initial attempt of a proof for the Hamming version, but  there is a problem, see below.}

The construction uses the Justessen concatenation scheme, which combines an \emph{outer code}, with several \emph{inner codes.} First the message written over the alphabet $\F_2^k$ (for some integer $k$)  is mapped by the encoder of the outer code  into a codeword over the same alphabet, and next each symbol of the codeword is mapped using the encoder of one the inner codes.

\emph{The outer code:}  Following Cheraghchi~\cite{che:c:codescondenser}, we use the linear time encodable/decodable  code constructed by Spielman~\cite{spi:j:code} as the outer code.
\begin{theorem}[~\cite{spi:j:code}]
For every constant $\alpha < 1$ and every positive integer $k$, there exist a constant $\beta_{\text{Spielman}} > 0$ and an explicit family of codes $(C_S)_{S \in \nat}$ over the alphabet $\F_2^k$, such that $C_S$ encodes messages of length $S$,   has rate $1-\alpha$,  and is error correcting for a fraction $\beta_{\text{Spielman}}$  of errors.The encoder and the decoder run in time that is linear in the bit-length of the codewords.  

More explicly,  for every $S \in \nat$, the encoder of $C_S$ maps $(\F_2^k)^S$ into $(\F_2^k)^D$, with $S/D \geq 1 - \alpha$,   and for every codeword $x \in (\F_2^k)^D$, and every $\tx$ such that that the relative Hamming distance between $x$ and $\tx$ is $\beta_{\text{Spielman}}$, the decoder on input $\tx$ returns $x$. The encoder and the decoder run in time $O(kD)$.
\end{theorem}

\emph{The inner codes:} We use $D = 2^d$ codes $\enc_{\rho_1}, \enc_{\rho_2}, \ldots, \enc_{\rho_D}$, which are obtained from the private code $\enc$ in Theorem~\ref{t:bruno}, by fixing the randomness to every $d$-bits string, i.e., for each $\rho \in \zo^d$,  $\enc_{\rho}(\cdot) \stackrel{def.}{=} \enc(\cdot, \rho)$. Each inner code maps a $k$-bit string, viewed in the natural way as an element of $\F_2^k$  into a $n$-bit string, which similarly is viewed in the natural way as an element of $\F_2^n$.

Equipped with  the outer code and the inner codes, we construct a new encoder $E$  as  follows.

\emph{Concatenation scheme:} First, the $S$-symbol input message  $(m_1, \ldots, m_S)$ is encoded with the outer code into a $D$-symbol word $(c_1, \ldots, c_D)$.
Next, each $c_i$ is encoded with the inner code $\enc_{\rho_i}$.
\[
 (m_1, \ldots, m_S)  \xrightarrow[\text{\quad outer code \quad }]{} (c_1, \ldots, c_D)   \xrightarrow[\text{\quad \quad  \quad inner codes \quad \quad \quad }]{} (\enc_{\rho_1}(c_1), \ldots,  \enc_{\rho_D}(c_D))
\]

Thus $E$ maps binary strings of length $k\cdot  S$ into codewords that are binary strings of length  $n \cdot D$.

Consider now a $D$-memoryless channel $\cha$ that has distortion bounded by $t$. 
Recall that this means that: 
\begin{enumerate}
\item  there is  a bipartite graph $G = (L=\zo^n, R, E \subseteq L \times R)$ and all right degrees are bounded by $T = 2^t$, and 
\item  $\cha$ takes as  input a $D$-tuple $(x_1, \ldots, x_D)$ of left nodes, and outputs $(y_1, \ldots, y_D)$, where each $y_i$ is chosen uniformly at random  among the right neighbors of $x_i$, and the $D$ random choices are independent.
\end{enumerate}

Suppose the sender transmits the codeword $(x_1, \ldots, x_D)$ (where each $x_i \in \F_2^n$) and the channel distorts it into $(y_1, \ldots, y_D)$. The decoder of the concatenated code, first calls the decoders of the $D$ inner codes respectively on each component of $(y_1, \ldots, y_D)$, which return the word $(z_1, \ldots, z_D)$, and next calls the decoder of the outer code on this latter word. We show that with high probability (over the randomness used by the channel), this procedure reconstructs $(x_1, \ldots, x_D)$.

For each right node $y$ of the graph $G$, the analysis in the proof of Theorem~\ref{t:bruno} shows that the event ``more than one codeword is in the neighborhood  of $y$" has probability $< \epsilon$ (the probability is over the randomness of the universal code in Theorem~\ref{t:bruno}). We say that $\rho \in \zo^d$ is \emph{good for $y$} if it satisfies the property in the above event. Thus, for each $y$, $(1-\epsilon)$ fraction of $\rho$'s are good. By a standard averaging argument, it follows that  there is a set (of ``good" $\rho$'s) GOOD$_{\text{rand}}$ containing $(1-\sqrt{\epsilon})$ fraction of $\rho$'s  in $\zo^d$ and  a set (of ``good" $y$'s)  GOOD$_{\text{node}}$  containing $(1-\sqrt{\epsilon})$   fraction of $y$'s  in $R$, such that every $\rho$ in GOOD$_{\text{rand}}$ is good for every $y$ in  GOOD$_{\text{node}}$. By rearranging the tuple $(y_1, \ldots, y_D)$, we can assume that the first $(1-\sqrt{\epsilon})D$ components correspond to the ``good" $\rho$'s.  On every  $i$ in this segment of ``good" components,  if  $y_i$ is ``good,"  the inner decoder is correct. Note that the probability (over the randomness of the channel)  that $y_i$ is ``good"  is at least $1 - \sqrt{\epsilon}$.  \marius{There's a problem here. It may happen that most (or even all) neighbors of $x_i$ are ``bad,"  and then the last sentence is wrong.  I think the argument can be salvaged for the oblivious scenario: Instead of looking at a right node $y$, we look at a noise vector $e$ in $E$, etc.} 

Let $\mu$ be  the expected number of components on which the inner decoders are incorrect. Note that 
\[
\mu \le \sqrt{\epsilon} \cdot D + (1-\sqrt{\epsilon}) \cdot D \cdot \sqrt{\epsilon}.  
\]
In the sum above, the first term  corresponds to the errors made by the inner decoders in the  bad segment and the second term corresponds to the errors made by decoders in the good segments for which the channel is choosing a bad $y$.  Thus, $\mu < \mu_H$, where $\mu_H \stackrel{def.}{=} 2 \sqrt{\epsilon}\cdot D$. We take $\gamma = \sqrt{\epsilon} \cdot D$.
\[
\prob [ \# \text{ errors } \ge 3 \sqrt{\epsilon} \cdot D] = \prob [ \# \text{ errors } \ge \mu_H + \gamma] \le e^{-2 \gamma^2/D} = e^{-2\epsilon \cdot D}. 
\]
 We have used a form of the Chernoff bounds for the case when we know an upperbound of the expected value (see Exercise 1.1 (a) in~\cite{dub-pan:b:concentration}).

Take $\epsilon$ such that $3 \sqrt{\epsilon} \le \beta_{\text{Spielman}}$. Then, with probability at least $1 -  e^{-2\epsilon \cdot D}$, the inner decoders are correct on all except at most  
$\beta_{\text{Spielman}}$ fraction of positions. In such a case, the outer decoder is able to reconstruct $(x_1, \ldots, x_D)$.

\fi

\if01
We show that in the oblivious scenario, to obtain a rate that is better than $1-t/n-o(1)$, we need at least a logarithmic amount of shared randomness.
For this, we consider private encoding functions that have access to both shared and non-shared randomness, i.e., we use encoding functions 
$\enc\colon \mcM \times \mcR \times \mcS$ where $\mcR$ represents shared randomness and~$\mcS$ non-shared randomness.

\begin{proposition}\label{prop:muchSharedRandomness}
  Let $\enc \colon \mcM \times \mcR \times \mcS \rightarrow \mcX$ be a private code with $\# \mcM = 3$.
  There exists a constant $\eps>0$ such that $\enc$ can only be $(\eps,t)$-resilient if $n \ge 1.49t$, where $n = \log \# \mcX$.
\end{proposition}

\noindent
We first show that private codes use at least $\Omega(\log t)$ randomness.
Since shared randomness is stronger than non-shared, we may assume the private functions only have shared randomness, 
and we drop the argument $\sigma \in \mcS$ from the encoding function.
We show that such codes use at least $\log T$ random bits. 

\begin{lemma}
  Assume $M = 2$ and $\eps < 1/2$. 
  If some private code $\enc \colon \mcM \times \mcR \rightarrow \mcX$ satisfies the above condition, then $\# \mcR \ge \log T$.
\end{lemma}

\begin{proof}
  Let $\enc$ be an encoding function that uses shared randomness from the set $\mcR = \{1,2,\ldots,r\}$. 
  Let $\mcM = \{\texttt{a}, \texttt{b}\}$.  
  Consider the channel defined by the set $E$ given by the span of the vectors
  \[
    v_1 = \enc(\texttt{a},1) - \enc(\texttt{b},1), \quad \ldots, \quad v_r = \enc(\texttt{a},r) - \enc(\texttt{b},r).
  \]
  Thus, $E$ has size at most~$2^r$.
  We need to select $m \in \mcM$ and $e \in E$ such that the probability in~\eqref{eq:oblivious} is at most~$1/2$.
  In the requirement~\eqref{eq:oblivious}, the only relevant values of $\dec$ are vectors of the form
  \[
     \enc_\rho(\texttt{a}) + c_1v_1 + \cdots + c_rv_r,
  \]
  with $\rho \in \mcR$ and $c \in \{0,1\}^r$.
  Select $\rho$ and $c$ randomly and consider the value of $\dec_\rho$ on the above vector, which is a value in~$\mcM$.
  Note that if we used message $\texttt{b}$ instead of $\texttt{a}$ in the expression above, 
  then the probabilities with which the messages appear does not change, (since this corresponds to flipping all bits of~$c$).
  Assume that the value $\texttt{b}$ appears at least half of the time. If this is not 
  the case, we flip the roles of $\texttt{a}$ and $\texttt{b}$ in the expression above and the explanations below. 
  There exists a choice of $c \in \{0,1\}^r$ such that for at least half of the values $\rho \in \mcR$, 
  the value of $\dec_\rho$ for the above vector is equal to~$\texttt b$.
  Let $e = c_1v_1 + \cdots + c_rv_r \in E$ be the corresponding vector. 
  For $m = \texttt{a}$, the probability in~\eqref{eq:oblivious} is at most~$1/2$. 
  Hence, for $\eps< 1/2$ the inequality is false, and this implies that if $\#\mcR \le \log T$ 
  equation~\eqref{eq:oblivious} can not be satisfied.
\end{proof}

\begin{remark}\label{rem:lowerbound}
The statement of the lemma is for $\eps$ arbitrarily close to~$1/2$.
If we only need a violation of  \eqref{eq:oblivious} with a smaller value of $\eps$, 
we can obtain a result for private codes
that use non-shared randomness drawn from a set $\mcS$ of size at most~$2^{-2\eps r}T$.
Here and below, we assume that $2\eps r$ is integer.
  Indeed, for a small error probability $\eps$ in \eqref{eq:oblivious}, 
it is enough to apply the reasoning to only a subset of $\mcR$ of size at most $2\eps r$, and this is 
requires at most $2^{2\eps r}$ error vectors. 
Hence, we can repeat the construction for all elements $\sigma \in \mcS$, 
and the channel $E$ is the union of all these error sets. 
This provides us with a channel with distortion $2^{2\eps r}\#\mcS$. 
Hence, if $\# \mcS \le 2^{-2\eps r}T$, equation~\eqref{eq:oblivious} can not be satisfied.
\end{remark}

\begin{proof}[Proof idea of Proposition~\ref{prop:muchSharedRandomness}.]
  Let $N = \#\mcX$. 
  Assume that at most $\log T$ shared randomness and $T$ is large. Also assume that $2\eps\log T$ is integer. 
  By remark~\ref{rem:lowerbound}, we conclude that either
  equation~\eqref{eq:oblivious} fails if~$\eps$ is small, or that we need~$\# \mcS > T^{1-2\eps}$. 
  In the first case, we are finished, so we may assume~$\# \mcS > T^{1-2\eps}$. 

  Let us now erroneously assume that $\enc_{\rho,\sigma}(m)$ always has different values for all $\sigma$. Then, we could apply 
  a sphere packing argument and conclude that $N \ge T^{2-2\eps}(1-\epsilon)$.
  But unfortunately, this might not be the case. From the argument of the lemma, 
  we can conclude that we have a significant amount of different sets of vectors $v_1, \ldots, v_r$ corresponding to 
  different choice of $\sigma \in \mcS$, but this does not exclude that we may have $\enc_{\rho,\sigma}(m)$ being a constant. 
  However, one can show that if 3 messages are available, then for 2 of them, the sets $\enc_{\rho,\mcS}$ have size at least~$T^{0.49}$.
  This can then be used in a sphere packing argument.
\end{proof}
\fi

\end{document}